\documentclass[reprint,amsmath,amssymb,aps,superscriptaddress]{revtex4-2}

\usepackage{graphicx}
\usepackage{braket}
\usepackage{amsthm}
\usepackage[caption=false]{subfig} 
\usepackage{appendix}
\usepackage{hyperref}
\hypersetup{colorlinks=true}
\usepackage{qcircuit}
\usepackage{tikz}
\newcommand{\refeq}[1]{Eq.~(\ref{#1})}
\newcommand{\refsec}[1]{Section~\ref{#1}}
\newcommand{\reffig}[1]{Figure~\ref{#1}}
\newcommand{\refapp}[1]{Appendix~\ref{#1}}

\theoremstyle{plain}
\newtheorem{theorem}{Theorem}
\newtheorem{lemma}[theorem]{Lemma}
\newtheorem{corollary}[theorem]{Corollary}
\theoremstyle{definition}
\newtheorem{definition}{Definition}
\newtheorem{problem}{Problem Statement}

\newcommand{\gc}{\mathrm{gc}}
\newcommand{\bp}{\mathrm{bp}}
\newcommand{\HIsing}{H_{\mathrm{Ising}}}
\newcommand{\norm}[1]{\|#1\|}
\newcommand{\Jij}{J_{i,j}} 
\newcommand \R{\mathbb{R}}
\DeclareMathOperator{\sgn}{sgn}
\newcommand{\Yb}{$^{171}\mathrm{Yb}^+$}

\usepackage{xcolor}
\usepackage[normalem]{ulem}
\let\isout\sout \renewcommand{\sout}[1]{\ifmmode\text{\isout{\ensuremath{#1}}}\else\isout{#1}\fi}

\newcommand{\tsout}[1]{} 
\newcommand{\change}[1]{#1} 

\newcommand{\rep}[2]{\tsout{#1}\change{#2}}

\begin{document}

\title{Generating Target Graph Couplings for QAOA from Native Quantum Hardware Couplings}

\author{Joel Rajakumar}
\altaffiliation{\change{Present address: University of Maryland, College Park, MD 20742, USA}}
\affiliation{Georgia Tech Research Institute, Atlanta, GA 30332, USA}

\author{Jai Moondra}
\affiliation{Georgia Institute of Technology, Atlanta, GA 30332, USA}

\author{\change{Bryan Gard}}
\affiliation{Georgia Tech Research Institute, Atlanta, GA 30332, USA}

\author{Swati Gupta}
\affiliation{Georgia Institute of Technology, Atlanta, GA 30332, USA}

\author{Creston D. Herold}
\email[]{Creston.Herold@gtri.gatech.edu}
\affiliation{Georgia Tech Research Institute, Atlanta, GA 30332, USA}

\date{\today}

\begin{abstract}
We present methods for constructing any target coupling graph using limited global controls in an Ising-like quantum spin system. Our approach is motivated by implementing the quantum approximate optimization algorithm (QAOA) on trapped ion quantum hardware to find approximate solutions to Max-Cut. We present a mathematical description of the problem and provide approximately optimal algorithmic constructions \rep{which}{that} generate arbitrary unweighted coupling graphs with $n$ nodes in $O(n)$ global \rep{control pulses}{entangling operations} and weighted graphs with $m$ edges in $O(m)$ \rep{pulses}{operations}. These upper bounds are not tight in general, and we formulate a mixed-integer program to solve the graph coupling problem to optimality. We perform numeric experiments on small graphs with $n\le8$ \change{and} show that optimal sequences, which use fewer operations, can be found using mixed-integer programs. \change{Noisy simulations of Max-Cut QAOA show that our implementation is less susceptible to noise than the standard gate-based compilation.}
\end{abstract}

\maketitle

\section{Introduction}
While the capabilities of quantum computing hardware continue to increase both in qubit number and quality of operations, the field is still a long way from quantum error correction. Numerous near-term applications for the so-called NISQ hardware have been proposed including quantum chemistry, the variational quantum eigensolver, and optimization \cite{preskill_quantum_2018}.
Important experimental demonstrations of these applications have been performed, but to date the hardware has been too small to be computationally useful compared to classical methods \cite{klco_quantum-classical_2018, nam_ground-state_2020, arute_hartree-fock_2020}.
Quantum hardware naturally performs better on problems which match the native qubit connectivity, while general problems tend to require significant compilation resulting in increased gate count and worse performance \cite{harrigan_quantum_2021}. 

A prominent example of an optimization problem is Max-Cut, which asks to partition the vertices of a given weighted graph $G = (V,E,z)$, with edge weights $z$, into two sets $S, V\setminus S$ so that the total weight of the edges in the cut $(S,V-S)$ is maximized.
To solve Max-Cut with the Quantum Approximate Optimization Algorithm (QAOA) \cite{farhi_quantum_2014, bravyi_obstacles_2020, farhi_quantum_2020},
each vertex is represented by a qubit and the $Z$-eigenstates label the two sets. The cut size is encoded as the energy of a cost Hamiltonian, consisting of a sum over two-qubit operators for each edge in $E$, where the Max-Cut corresponds to the lowest energy eigenvector.
Max-Cut QAOA was recently demonstrated on tens of ions \cite{pagano_quantum_2020}, however, the problem graph matched the native hardware coupling graph. To implement Max-Cut for an any graph, an arbitrary qubit coupling graph (cost Hamiltonian) must be produced.

In this paper, we will show how arbitrary coupling operations can be constructed on a quantum spin system using a limited set of operations.
While we specialize to a trapped-ion system, our method is general and is applicable to other quantum systems with natural long-range interactions like NMR manipulation of multi-spin molecules \cite{Vandersypen2001}, arrays of Rydberg atoms \cite{levine2018, graham2019, Henriet2020}\change{, or certain superconducting qubits coupled to a common bus resonator \cite{xu_emulating_2018,onodera_quantum_2020}.}
While our results have immediate application to efficient implementation of Max-Cut QAOA on crystals of trapped ions, they may also enable new classes of quantum simulations.

In \tsout{section} \refsec{sec:background}, we  provide a physical definition of a coupling graph as well as implementation details for \tsout{the} trapped ion systems.
In \refsec{sec:problem}, we physically motivate the problem and then give a problem statement in purely mathematical terms. 
In \refsec{sec:solution}, we discuss various solution methods, including (i) a lower bound of $\Omega(\log n)$
\footnote{We use standard notation for limiting behavior of real-valued functions. Given functions $f, g: \mathbb{N} \rightarrow \mathbb{R}$, we say $f(n) = O(g(n))$ if there exists a positive real $M$ and a positive integer $n_0$ such that $|f(x)| \leq M\:|g(x)|$ for all $n\geq n_{0}$, and we say that $g(n) = \Omega(f(n))$ if $f(n) \in O(g(n))$.} on the graph coupling number, i.e., smallest possible number of \rep{pulses}{Ising operations} required to construct a general graph with $n$ nodes, (ii) an $O(n)$ \rep{pulse}{operation} sequence for the  construction of unweighted graphs which we refer to as {\it union-of-stars}, and (iii) an $O(m)$ \rep{pulse}{operation} for constructing arbitrary weighted graphs with $m$ edges.

\rep{Finally, p}{P}reliminary numerical experiments using a mixed integer program (MIP) are presented in \refsec{sec:results} which validate the performance of our constructions and show that optimal sequences can be found which use fewer \rep{pulses}{Ising operations} than the upper bounds from the constructions. This MIP can be used to find optimal \tsout{pulse} sequences for small instances (up to the scale of the NISQ hardware in the near future), but for large graphs this approach quickly becomes intractable, and instead, the union-of-stars construction can be used to find a provable near-optimal solution. \change{Numerical simulations of Max-Cut QAOA including noise follow in \refsec{sec:noisy}. Finally, an estimate of hardware execution time is given in \refsec{sec:runtime}.}

\section{Coupling Operators and Physical Controls}\label{sec:background}
Many quantum systems exhibit intrinsic couplings between qubits. Physically, these spin-spin couplings comprise the zero field Ising model with coupling Hamiltonian
\begin{equation}
    \label{eq:HIsing}
    \HIsing = \sum_{i=1}^{n-1} \sum_{j=i+1}^{n} \Jij \sigma_i^z \sigma_j^z \,,
\end{equation}
where $\sigma_j^z$ denotes the Pauli-Z matrix acting on the $j$th ion \footnote{Explicitly, this is defined on an $n$-qubit system by the tensor product $\sigma_i^z = \mathbb{I}^{\otimes i-1} \otimes \sigma^z \otimes \mathbb{I}^{\otimes n-i}$, i.e. an embedding of the $2 \times 2$ Pauli-Z matrix in $2^n \times 2^n$ space.}. The term $\sigma_i^z \sigma_j^z$ is a ``$ZZ$-coupling'' between qubits $i$ and $j$, and the values of the \change{interaction energies} $\Jij$ depend on the details of the underlying quantum hardware.

\tsout{We define an arbitrary coupling operator as}
\begin{equation*}
    \tsout{C = \sum_{i=1}^{n-1} \sum_{j=i+1}^{n} a_{i,j} \sigma_i^z \sigma_j^z}
\end{equation*}
\tsout{where $a_{i,j} \in \mathbb{R}$ and $\sigma_i^z \sigma_j^z$ is a unit-strength coupling between qubits $i$ and $j$. In this paper, we will use an abstraction of the coupling operation called the coupling graph, as defined below.}

\begin{definition} (Coupling Graph)
We define a coupling graph $G(V,E), |V| = n$ as a mathematical abstraction of a coupling operation on an $n$-qubit system, where each vertex $v \in V$ represents a qubit and each edge $e \in E$ represents the strength of the $ZZ$-coupling between two qubits. \tsout{The coupling operation given by \refeq{eq:coupling} produces a coupling graph with adjacency matrix $A$ of entries $a_{i,j} = a_{j,i}$.}
\end{definition}

\change{With this definition, an arbitrary coupling operator
\begin{equation}
    \label{eq:coupling}
    C = \sum_{i=1}^{n-1} \sum_{j=i+1}^{n} a_{i,j} \sigma_i^z \sigma_j^z
\end{equation}
has a coupling graph with adjacency matrix $A$ of entries $a_{i,j} = a_{j,i}$, where $a_{i,j} \in \mathbb{R}$.}
$\HIsing$ is a particular physical coupling operator (compare to \refeq{eq:coupling}), and the adjacency matrix of the corresponding coupling graph has elements $a_{i,j} = \Jij$.

In \refsec{sec:problem}, we show that a coupling operator $C$ with any arbitrary adjacency matrix can be constructed from two simple operations: (A) a global coupling operation, $\HIsing$, native to the quantum hardware, and (B) individual qubit bit flips. Here, we specialize to collections of ions trapped in a common potential and describe how each of these operations is realized.

\subsection{$\HIsing$ Implementation}

For trapped ion crystals, a wide variety of coupling graphs, facilitated by the collective normal modes of motion of the ions, are possible.
One method for generating $\HIsing$ is the M{\o}lmer-S{\o}rensen (MS) interaction \cite{Sorensen2000}.
With one-dimensional ion crystals, the MS interaction has been used extensively for quantum simulations of interacting spins \cite{zhang_observation_2017,kaplan_many-body_2020} as well as lattice gauge models \cite{kokail_self-verifying_2019}.
For experiments to date with two-dimensional ion crystals, an optical dipole force (ODF) was used to produce $\HIsing$ \cite{britton_engineered_2012,bohnet_quantum_2016}.
While the exact dependence of the interaction strength on the optical fields is different, the dependence of $\Jij$ on the motional mode structure is identical. 

The MS interaction is created by illuminating the ions with a pair of lasers tuned near red and blue motional sidebands, and the $\Jij$ in \refeq{eq:HIsing} are
\begin{equation}\label{eq:MS-Jij}
		\Jij=\Omega_i \Omega_j \frac{\hbar^2 (\Delta k)^2}{2M} \sum_m \frac{b_{i,m} b_{j,m}}{\mu^2 - \omega_m^2} \,,
\end{equation}
where $M$ is the single-ion mass, $\Delta k$ is the momentum imparted by the laser interaction, $\hbar$ is the reduced Planck constant, and $\Omega_j$ is the Rabi rate for the $j$th ion, which depends on atomic matrix elements and the optical intensity at the ion \cite{kim_entanglement_2009}. 
Within the sum over motional modes $m$, the normal mode frequencies have angular frequency $\omega_m$, and the $b_{i,m}$ are the displacement amplitudes for ion $i$ in mode $m$ (see, e.g., \cite{James1998}); $\mu$ is the angular frequency of the laser with respect to the qubit frequency.
The natural basis for the MS interaction can be taken to be $\sigma_i^x \sigma_j^x$, however it can be transformed to $\sigma_i^z \sigma_j^z$ through global rotations with the same pair of laser beams \cite{Lee2005,pino_demonstration_2021}.

We define the \rep{gate}{MS} detuning as the difference between the laser detuning and the $m$th normal mode, i.e.:
\begin{equation}
    \delta_m = \mu-\omega_m.
\end{equation}
While a wide variety of $\Jij$ are possible through variation of the laser intensities, the \rep{gate}{MS} detuning, and the motional modes \cite{islam_quantum_2012}, here we focus on two specific cases. The $m=0$ mode is the center of mass mode for which all the $b_{j,0}$ are equal. 
If the MS interaction is detuned very close to this mode such that $|\delta_0| \ll \delta_m ~\forall m\ne0$, the $m=0$ term in the sum in \refeq{eq:MS-Jij} dominates, and we have:
\begin{equation}
    \label{eq:constantJ}
    \Jij = \sgn\left(\delta_0\right) J_0, 
\end{equation}
for positive constant $J_0$.
In the language of spin models of magnetism, this produces an infinite range ferromagnetic or antiferromagnetic interaction depending on the sign of the detuning \cite{kim_entanglement_2009}. In particular, \rep{in this work, we assume that $J_{i,j}$ all identically equal to either 1 or -1 for $i \neq j$ and $J_{i,j}=0$ otherwise.}{we show in \refsec{sec:uos} that this equal, all-to-all interaction permits efficient generation of dense, unweighted coupling graphs.} 

\subsection{Bit Flip Implementation}
One-qubit operations can be envisioned as rotations of a qubit vector on the Bloch sphere. The unitary $R_j^x(\theta) = \exp\{-i(\theta/2) \sigma_j^x\}$ is a rotation through an angle $\theta$ about the $x$-axis. We define a bit flip on the $j$th qubit as:
\begin{equation}\label{eq:bitflip}
    X_j = R_j^x(\pi) = -i \sigma_j^x.
\end{equation}
$X_j$ will flip the state of a qubit in the $z$- or $y$-basis.
We assume that these bit flips can be applied to an individual ion within the ion crystal, which is accomplished through tightly focused control beams.
Such rotations can be performed sequentially or in parallel on many ions since they do not change the motional state of the crystal.

\section{Problem Statement}\label{sec:problem}
Arbitrary coupling graphs can be composed from the hardware-native $\HIsing$ global interaction and single-qubit bit flips introduced in \refsec{sec:background}.
Following a physical motivation, we give a precise mathematical description of this construction problem which does not require any background physical knowledge.

\subsection{Physical Description}
Motivated by solving Max-Cut with QAOA \cite{farhi_quantum_2014}, we want to find a set of physical operations which efficiently implements the cost Hamiltonian for a particular target graph.
The cost Hamiltonian is exactly the coupling operation defined in \refeq{eq:coupling} apart from overall constants.
This is similar to any quantum operation compilation problem, where a desired unitary is constructed from a limited set of physical operations. We can succinctly describe the physical problem:
\begin{problem} [Physical]
    Find an ``\rep{pulse}{operation} sequence'' composed by interleaving ${\HIsing}$ interactions and single-qubit bit flips $X_j$, that produces a desired coupling operation. To mitigate operation error, a short sequence is desired; we will optimize for minimum number of $\HIsing$ applications or the minimum total strength of $\HIsing$ operations.
\end{problem}

The intuition behind this limited set of controls is as follows. By surrounding a $ZZ$ term with bit flips, the sign of that term in $\HIsing$ can be reversed:
\begin{equation}
    X_j^\dagger \sigma_i^z \sigma_j^z X_j = -\sigma_i^z \sigma_j^z \,.
\end{equation}
However, $X_j^\dagger \HIsing X_j$ will change the sign of all $ZZ$ terms acting on ion $j$.
Repeated applications of $\HIsing$ surrounded by carefully chosen sequences of bit flips were previously used to isolate individual two-qubit interactions within molecules in NMR quantum computing \cite{Leung2000}.
\change{Similar controls were also proposed to generate four-body spin interactions \cite{muller_simulating_2011}.}

\tsout{The coupling for the $p$th application of $\HIsing$ with strength $w_p$ is:}
\rep{and the $P_{p,i}$ are elements of a ``pulse matrix'' $P$ with values $\pm1$. The pulse matrix has a column for each qubit and a row for each application of $\HIsing$.}{We represent these bit flip operations with a $\{\pm 1\}^{k\times n}$ matrix $P$, where $k$ is the number of $\HIsing$ applications, and there is a column for each qubit.} The sign of each entry in a row encodes the bit flips surrounding that application of $\HIsing$\rep{. If there is a $-1$ in a given column the corresponding qubit is flipped, and for a $1$ it is not.}{, i.e., $P_{p, i}=-1$ if $i$th qubit is flipped in the $p$th application of $\HIsing$, and is equal to 1 otherwise. The coupling for the $p$th application of $\HIsing$ with strength $w_p$ can then be written as:}
\begin{align}
    \label{eq:pulse-coupling}
    C^{(p)} &= \sum_{i=1}^{n-1} \sum_{j=i+1}^n c_{p,i,j} \sigma_i^z \sigma_j^z \,,
    \\
    & \mbox{where } \, c_{p,i,j} = w_p \Jij P_{p,i} P_{p,j} \,. \nonumber
\end{align}
This can been \rep{seen}{verified} by considering the four possible bit flips for the $i,j$ term:
\begin{subequations}
\begin{align}
&w_p J_{i,j}\sigma_i^z \sigma_j^z,
\nonumber \\
P_{p,i} &= P_{p,j} = 1,\\
X_j^{\dagger} (w_p J_{i,j}\sigma_i^z \sigma_j^z) X_j = - &w_p J_{i,j}\sigma_i^z \sigma_j^z,
\nonumber \\
P_{p,i} &= 1, P_{p,j} = -1,\\
X_i^{\dagger} (w_p J_{i,j}\sigma_i^z \sigma_j^z) X_i = -&w_p J_{i,j}\sigma_i^z \sigma_j^z,
\nonumber \\
P_{p,i} &= -1, P_{p,j} = 1,\\
X_j^{\dagger} X_i^{\dagger} (w_p J_{i,j}\sigma_i^z \sigma_j^z) X_i X_j =& w_p J_{i,j}\sigma_i^z \sigma_j^z,
\nonumber \\
P_{p,i} &= P_{p,j} = -1 \,.
\end{align}
\end{subequations}
The net coupling of a series of $k$ applications of $\HIsing$ is given by the sum of their couplings:
\begin{equation}
    C = \sum_{p=1}^k C^{(p)} 
    = \sum_{i=1}^{n-1} \sum_{j=i+1}^n \left(\sum_{p=1}^k c_{p,i,j}\right) \sigma_i^z \sigma_j^z \,.
\end{equation}

By comparison with \refeq{eq:coupling}, we identify that the term in parenthesis is an element of the total coupling adjacency matrix, which can therefore be expressed as
\begin{equation}
    A = P^T W P \odot J \,,
\end{equation}
where $W$ is a diagonal $k\times k$ matrix of \rep{pulse}{Ising interaction} strengths with $W_{p,p} = w_p$, $\odot$ represents element wise matrix multiplication, and \tsout{the pulse matrix} $P \in \{ 1, -1\}^{k\times n}$.

Both positive and negative strengths $w_p$ are possible. For $\Jij = +|J_0|$, we can choose the strength of $\HIsing$ by changing the intensity of the lasers creating the MS interaction; however, this only permits $w_p > 0$. To achieve an effective $w_p < 0$, the MS interaction detuning can be negated for $\Jij = -|J_0|$ through \refeq{eq:constantJ}.
In \refapp{app:example}, we give an example of an \rep{pulse}{operation} sequence which creates a particular three-qubit coupling graph.

An optimal \rep{pulse}{operation} sequence is one which produces the desired coupling graph for the target problem with the least total error. It is expected that the infidelity of $\HIsing$ will be much larger than for the $X_j$ bit flips.
We propose two optimization criteria: (i) minimizing the total number of $\HIsing$ applications, and (ii) minimizing the total absolute strength of the $\HIsing$ \rep{pulses}{operations}. The first accounts for errors which are fixed for every entangling operation, while the second criterion will account for strength-dependent effects like residual spin-motion entanglement.
In the mathematical description that follows, these two criteria are expressed as the $L_0$ and $L_1$ norms of $W$, respectively.

\subsection{Mathematical Description}
We present in this section a mathematical statement of the \rep{graph compilation}{target coupling graph construction} problem, as motivated \rep{from the experimental setup}{by the preceding physical description}:
\begin{problem}[Mathematical]
    Given the target adjacency matrix $A \in \mathbb{R}^{n \times n}$ of graph $G$ and an intrinsic (hardware-dependent) adjacency matrix $J \in \mathbb{R}^{n \times n}$ for $n \in \mathbb{Z}_+$, find $P \in \{1,-1\}^{k \times n}$ and 
    a diagonal strength matrix $W \in \mathbb{R}^{k \times k}$ such that 
    \begin{equation} \label{eq:MathStatement}
        A = P^T W P \odot J
    \end{equation}
    where $\odot$ is the element-wise multiplication operation. We \rep{want to minimize either}{are interested in minimizing two objectives:} $\norm{W}_{0}$ \rep{or}{and} $\norm{W}_{1}$, which denote the $L_0$ and $L_1$ norms of the diagonal of $W$, respectively. \rep{That is}{In other words}, $\norm{W}_0$ is the number of non-zero entries along the diagonal, and $\norm{W}_1$ is the sum of the absolute values of the diagonal entries.
\end{problem}

Motivated by the discussion leading to \refeq{eq:constantJ}, for the rest of the paper we restrict ourselves to $J$ corresponding to the unweighted complete graph $K_n$. Let the strength matrix $W$ take arbitrary values, and fix $J$ to be $J_{i,j} = 1$ for $i \neq j$, and $J_{i,i}=0$ for all $i \in \{1, \hdots, n\}$.

\rep{In particular we are interested in}{We will first consider} minimizing the number of \rep{pulses}{Ising operations}, i.e., the $L_0$ norm of the strength matrix $W$, and define a new \rep{coupling number of the graph accordingly}{combinatorial quantity called the {\it coupling number} of a graph}:

\begin{definition}(Graph Coupling number)
We denote the minimum number of \rep{pulses}{Ising operations} needed to construct a graph, i.e., the $L_0$ norm $\|W\|_0$, as the graph coupling number, denoted by $\gc(G)$.
\end{definition}

\rep{We are not aware of any previous combinatorial optimization or algorithmic formulation for this problem, however it has interesting connections to the biclique partition number. Given a graph $G$, a biclique in the graph is a complete bipartite subgraph, and a clique is a complete subgraph. In \refsec{sec:solution}, we give an approximation algorithm for this problem that relies on partitioning the graph edges into bicliques. The size of a minimum such partition is called the \emph{biclique partition number}, defined as follows:}{This problem has not been studied in combinatorial optimization (to the best of our knowledge). We first observe that any $\{\pm 1\}$ operation sequence creates a sign pattern over a complete graph: $-1$ edges corresponding to a complete bipartite graph, and $+1$ edges correspond to complete graphs within each bipartition. Therefore, the notion of a graph coupling number has an interesting connection to the following \emph{biclique partition number}}:

\begin{definition}\label{biclique-partition-number}
Given an undirected graph $G$, the \emph{biclique partition number} of graph $G$ is defined as the minimum number of edge-disjoint bicliques (i.e., complete bipartite graphs) of $G$ whose union includes all of the edges of $G$. We denote this number by $\bp(G)$.
\end{definition}

The problem of determining the biclique partition number of a graph was introduced by Graham and Pollak in 1971 \cite{graham_pollak_loop_switching} and is known to be NP-hard \cite{orlin_1977}. We will show that one can construct any biclique using a constant number of \rep{pulses}{Ising operations} (irrespective of the size of the graph), which can be used to show $\gc(G) \leq 3\bp(G)+1$ for any unweighted graph $G$. However, this does not immediately imply NP-hardness for the graph coupling problem -- indeed, $\bp(K_n) = n - 1$ \cite{graham_pollak_loop_switching} whereas $\gc(K_n) = 1$. \rep{In this work, we}{We next} lay down some initial observations and approximation bounds for the graph coupling number \change{(using a partitioning of the graph edges into bicliques, in \refsec{sec:solution}).} 
We verify experimentally that our worst-case approximation bounds are order-optimal for small graphs (modulo the choice of a parameter $M$, defined in \refapp{app:MIP}), and therefore, these methods can be of immediate use to quantum physicists.

\section{Solution Methods}\label{sec:solution}
In this section, we give combinatorial methods for constructing \rep{pulse}{operation} sequences with small $L_0$ norm. Our proposed construction of unweighted graphs has a linear upper bound on the $L_0$ norm, i.e., $\norm{W}_0 = O(n)$, and our construction for weighted graphs requires $\norm{W}_0 = O(m)$ number of \rep{pulses}{Ising operations}. We also show that any graph with distinct edge weights needs at least $\Omega(\log n)$ \rep{pulses}{Ising operations}, i.e., the maximum graph coupling number on all graphs with $n$ vertices is $\gc(G) = \Omega(\log n)$. Finally, we describe a simple mixed-integer program to find optimal solutions on small graphs.

Prior work by \citet{Leung2000} outlined a deterministic method to couple any desired pair of qubits within the molecules used for NMR quantum computing.
NMR qubits are subject to a global, pairwise $ZZ$-coupling, like $\HIsing$, which is always on. They proposed a construction to decouple all qubits and selectively recouple any single pair of qubits in $O(n)$ \rep{$\HIsing$ pulses}{Ising operations} interleaved with single-qubit bit flips, which can be mapped to the problem formulation in \refeq{eq:MathStatement} with strictly non-negative $W$. Therefore, using this construction for all $m = |E|$ edges in $E$, a target coupling graph can be produced in $O(n m)$ \rep{$\HIsing$ pulses}{Ising operations}.
In our trapped ion system, however, $\HIsing$ strengths can also be negative, and therefore, we are able to simplify the construction as well as reduce the upper bound on the number of \rep{pulses}{operations} to $O(m)$ for weighted graphs and $O(n)$ for unweighted graphs.

\subsection{Union-of-Stars \change{Construction}}\label{sec:uos}
In this section, we give general constructions for weighted and unweighted graphs. We will construct a weighted graph in $O(m) = O(n^2)$ \rep{pulses}{Ising operations} by constructing each edge in a constant number of steps and composing those constructions. For unweighted graphs, we will construct a star subgraph (wherein \change{a} single node is adjacent to a set of non-adjacent nodes) in a constant number of steps and compose these constructions to build the graph. Since the edge set of an unweighted graph is the union of at most $n - 1$ edge-disjoint stars, this results in an $O(n)$ construction for unweighted graphs. \change{Two example union-of-stars constructions for the same graph are given in \reffig{fig:union-of-stars-example-1} and \reffig{fig:union-of-stars-example-2}.}

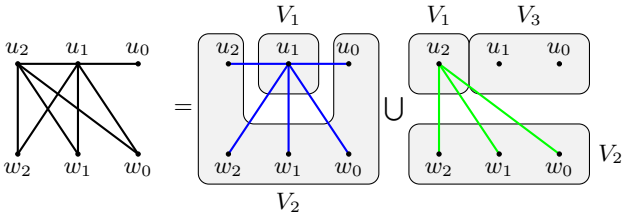
\begin{figure}
\begin{tikzpicture}[scale=0.40]
    \draw[black, thick] (-1,2) -- (3,2);   
    \draw[black, thick] (-1,2) -- (-1,-1); 
    \draw[black, thick] (-1,2) -- (1,-1);  
    \draw[black, thick] (-1,2) -- (3,-1);  
    \draw[black, thick] (1,2) -- (-1,-1);  
    \draw[black, thick] (1,2) -- (1,-1);   
    \draw[black, thick] (1,2) -- (3,-1);   
    
    \filldraw[black] (-1,2) circle (2pt) node[anchor=south] (u2_0) {$u_2$};
    \filldraw[black] (1,2) circle (2pt) node[anchor=south] (u1_0) {$u_1$};
    \filldraw[black] (3,2) circle (2pt) node[anchor=south] (u0_0) {$u_0$};
    \filldraw[black] (-1,-1) circle (2pt) node[anchor=north] (w2_0) {$w_2$};
    \filldraw[black] (1,-1) circle (2pt) node[anchor=north] (w1_0) {$w_1$};
    \filldraw[black] (3,-1) circle (2pt) node[anchor=north] (w0_0) {$w_0$};

    \node[align=center] at (4.5,0.5) {$=$};

    \draw[fill=gray!10][rounded corners] (7,3)--(9,3)--(9,1)--(7,1)--cycle;
    \filldraw[black] (8,3) circle (0pt) node[anchor=south]{$V_1$};
    \draw[fill=gray!10][rounded corners] (11,3)--(11,-2)--(5,-2)--(5,3)--(6.5,3)--(6.5,0)--(9.5,0)--(9.5,3)--cycle;
    \filldraw[black] (8,-2) circle (0pt) node[anchor=north]{$V_2$};

    \draw[blue, thick] (6,2) -- (10,2);   
    \draw[blue, thick] (8,2) -- (6,-1);  
    \draw[blue, thick] (8,2) -- (8,-1);   
    \draw[blue, thick] (8,2) -- (10,-1);   
    
    \filldraw[black] (6,2) circle (2pt) node[anchor=south]{$u_2$};
    \filldraw[black] (8,2) circle (2pt) node[anchor=south]{$u_1$};
    \filldraw[black] (10,2) circle (2pt) node[anchor=south]{$u_0$};
    \filldraw[black] (6,-1) circle (2pt) node[anchor=north]{$w_2$};
    \filldraw[black] (8,-1) circle (2pt) node[anchor=north]{$w_1$};
    \filldraw[black] (10,-1) circle (2pt) node[anchor=north]{$w_0$};
    
    \node[align=center] at (11.5,0.5) {$\bigcup$};
    
    \draw[fill=gray!10][rounded corners] (12,3)--(14,3)--(14,1)--(12,1)--cycle;
    \filldraw[black] (13,3) circle (0pt) node[anchor=south]{$V_1$};
    \draw[fill=gray!10][rounded corners] (12,0)--(18,0)--(18,-2)--(12,-2)--cycle;
    \filldraw[black] (18,-1) circle (0pt) node[anchor=west]{$V_2$};
    \draw[fill=gray!10][rounded corners] (14,3)--(18,3)--(18,1)--(14,1)--cycle;
    \filldraw[black] (16,3) circle (0pt) node[anchor=south]{$V_3$};
    
    \draw[green, thick] (13,2) -- (13,-1); 
    \draw[green, thick] (13,2) -- (15,-1);  
    \draw[green, thick] (13,2) -- (17,-1);  
    
    \filldraw[black] (13,2) circle (2pt) node[anchor=south]{$u_2$};
    \filldraw[black] (15,2) circle (2pt) node[anchor=south]{$u_1$};
    \filldraw[black] (17,2) circle (2pt) node[anchor=south]{$u_0$};
    \filldraw[black] (13,-1) circle (2pt) node[anchor=north]{$w_2$};
    \filldraw[black] (15,-1) circle (2pt) node[anchor=north]{$w_1$};
    \filldraw[black] (17,-1) circle (2pt) node[anchor=north]{$w_0$};
\end{tikzpicture}

\caption{\change{An example construction for the union-of-stars algorithm for an unweighted graph (Theorem \ref{unweighted}). The original graph (left) can be decomposed into two star graphs, represented in blue and green, respectively. Since each star is a complete bipartite graph, it can be constructed using Lemma \ref{complete-bipartite-construction}, with the sets $V_1, V_2, V_3$ specified for each of the stars. The construction of each star takes $4$ Ising operations, but since one operation is common across the stars, we end up with a total of $7$ operations for this construction by combining the common one using Lemma \ref{redundant-rows}.}}
\label{fig:union-of-stars-example-1}
\end{figure}

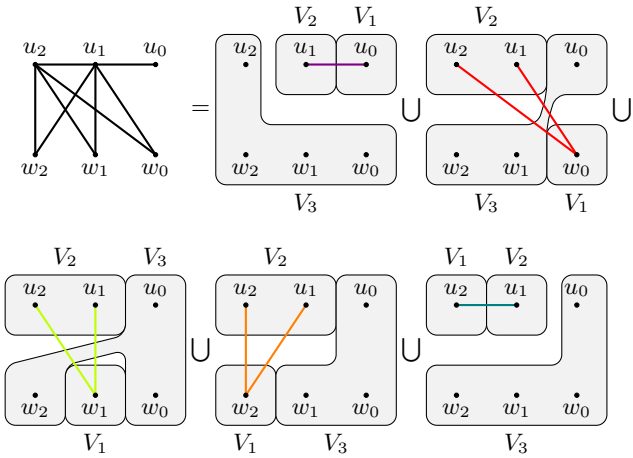
\begin{figure}
\begin{tikzpicture}[scale=0.40]
    \draw[black, thick] (-1,2) -- (3,2);   
    \draw[black, thick] (-1,2) -- (-1,-1); 
    \draw[black, thick] (-1,2) -- (1,-1);  
    \draw[black, thick] (-1,2) -- (3,-1);  
    \draw[black, thick] (1,2) -- (-1,-1);  
    \draw[black, thick] (1,2) -- (1,-1);   
    \draw[black, thick] (1,2) -- (3,-1);   
    
    \filldraw[black] (-1,2) circle (2pt) node[anchor=south]{$u_2$};
    \filldraw[black] (1,2) circle (2pt) node[anchor=south]{$u_1$};
    \filldraw[black] (3,2) circle (2pt) node[anchor=south]{$u_0$};
    \filldraw[black] (-1,-1) circle (2pt) node[anchor=north]{$w_2$};
    \filldraw[black] (1,-1) circle (2pt) node[anchor=north]{$w_1$};
    \filldraw[black] (3,-1) circle (2pt) node[anchor=north]{$w_0$};

    \node[align=center] at (4.5,0.5) {$=$};

    \draw[fill=gray!10][rounded corners] (9,3)--(11,3)--(11,1)--(9,1)--cycle;
    \filldraw[black] (10,3) circle (0pt) node[anchor=south]{$V_1$};
    \draw[fill=gray!10][rounded corners] (11,-2)--(5,-2)--(5,3)--(6.5,3)--(6.5,0)--(11,0)--cycle;
    \filldraw[black] (8,-2) circle (0pt) node[anchor=north]{$V_3$};
    \draw[fill=gray!10][rounded corners] (7,3)--(9,3)--(9,1)--(7,1)--cycle;
    \filldraw[black] (8,3) circle (0pt) node[anchor=south]{$V_2$};
    
    \draw[violet, thick] (8,2) -- (10,2);  \filldraw[black] (6,2) circle (2pt) node[anchor=south]{$u_2$};
    \filldraw[black] (8,2) circle (2pt) node[anchor=south]{$u_1$};
    \filldraw[black] (10,2) circle (2pt) node[anchor=south]{$u_0$};
    \filldraw[black] (6,-1) circle (2pt) node[anchor=north]{$w_2$};
    \filldraw[black] (8,-1) circle (2pt) node[anchor=north]{$w_1$};
    \filldraw[black] (10,-1) circle (2pt) node[anchor=north]{$w_0$};
    
    \node[align=center] at (11.5,0.5) {$\bigcup$};
    
    \draw[fill=gray!10][rounded corners] (18,-2)--(16,-2)--(16,0)--(18,0)--cycle;
    \filldraw[black] (17,-2) circle (0pt) node[anchor=north]{$V_1$};
    \draw[fill=gray!10][rounded corners] (12,3)--(12,1)--(16,1)--(16,3)--cycle;
    \filldraw[black] (14,3) circle (0pt) node[anchor=south]{$V_2$};
    \draw[fill=gray!10][rounded corners] (12,0)--(12,-2)--(16,-2)--(16,0)--(16.3,1)--(18,1)--(18,3)--(16,3)--(16,1)--(15.7,0)--cycle;
    \filldraw[black] (14,-2) circle (0pt) node[anchor=north]{$V_3$};
    
    \draw[red, thick] (13,2) -- (17,-1);  
    \draw[red, thick] (15,2) -- (17,-1);   
    
    \filldraw[black] (13,2) circle (2pt) node[anchor=south]{$u_2$};
    \filldraw[black] (15,2) circle (2pt) node[anchor=south]{$u_1$};
    \filldraw[black] (17,2) circle (2pt) node[anchor=south]{$u_0$};
    \filldraw[black] (13,-1) circle (2pt) node[anchor=north]{$w_2$};
    \filldraw[black] (15,-1) circle (2pt) node[anchor=north]{$w_1$};
    \filldraw[black] (17,-1) circle (2pt) node[anchor=north]{$w_0$};
    
    \node[align=center] at (18.5,0.5) {$\bigcup$};
    
    \draw[fill=gray!10][rounded corners] (0,-8)--(0,-10)--(2,-10)--(2,-8)--cycle;
    \filldraw[black] (1,-10) circle (0pt) node[anchor=north]{$V_1$};
    \draw[fill=gray!10][rounded corners] (-2,-5)--(-2,-7)--(2,-7)--(2,-5)--cycle;
    \filldraw[black] (0,-5) circle (0pt) node[anchor=south]{$V_2$};
    \draw[fill=gray!10][rounded corners]
    (2,-7.5)--(2,-10)--(4,-10)--(4,-5)--(2,-5)--(2,-7)--(-2,-8)--(-2,-10)--(0,-10)--(0,-8)--cycle;
    \filldraw[black] (3,-5) circle (0pt) node[anchor=south]{$V_3$};
    
    \draw[lime, thick] (-1,-6) -- (1,-9);  
    \draw[lime, thick] (1,-6) -- (1,-9);   
    
   \filldraw[black] (-1,-6) circle (2pt) node[anchor=south]{$u_2$};
    \filldraw[black] (1,-6) circle (2pt) node[anchor=south]{$u_1$};
    \filldraw[black] (3,-6) circle (2pt) node[anchor=south]{$u_0$};
    \filldraw[black] (-1,-9) circle (2pt) node[anchor=north]{$w_2$};
    \filldraw[black] (1,-9) circle (2pt) node[anchor=north]{$w_1$};
    \filldraw[black] (3,-9) circle (2pt) node[anchor=north]{$w_0$};

    
    \node[align=center] at (4.5,-7.5) {$\bigcup$};

    \draw[fill=gray!10][rounded corners] (5,-8)--(5,-10)--(7,-10)--(7,-8)--cycle;
    \filldraw[black] (6,-10) circle (0pt) node[anchor=north]{$V_1$};
    \draw[fill=gray!10][rounded corners] (5,-5)--(5,-7)--(9,-7)--(9,-5)--cycle;
    \filldraw[black] (7,-5) circle (0pt) node[anchor=south]{$V_2$};
    \draw[fill=gray!10][rounded corners] (9,-5)--(9,-8)--(7,-8)--(7,-10)--(11,-10)--(11,-5)--cycle;
    \filldraw[black] (9,-10) circle (0pt) node[anchor=north]{$V_3$};
    
    \draw[orange, thick] (6,-6) -- (6,-9); 
    \draw[orange, thick] (8,-6) -- (6,-9);  
    
    \filldraw[black] (6,-6) circle (2pt) node[anchor=south]{$u_2$};
    \filldraw[black] (8,-6) circle (2pt) node[anchor=south]{$u_1$};
    \filldraw[black] (10,-6) circle (2pt) node[anchor=south]{$u_0$};
    \filldraw[black] (6,-9) circle (2pt) node[anchor=north]{$w_2$};
    \filldraw[black] (8,-9) circle (2pt) node[anchor=north]{$w_1$};
    \filldraw[black] (10,-9) circle (2pt) node[anchor=north]{$w_0$};
    
    \node[align=center] at (11.5,-7.5) {$\bigcup$};
    
    \draw[fill=gray!10][rounded corners] (12,-5)--(12,-7)--(14,-7)--(14,-5)--cycle;
    \filldraw[black] (13,-5) circle (0pt) node[anchor=south]{$V_1$};
    \draw[fill=gray!10][rounded corners] (16,-5)--(16,-7)--(14,-7)--(14,-5)--cycle;
    \filldraw[black] (15,-5) circle (0pt) node[anchor=south]{$V_2$};
    \draw[fill=gray!10][rounded corners] (16.5,-5)--(18,-5)--(18,-10)--(12,-10)--(12,-8)--(16.5,-8)--cycle;
    \filldraw[black] (15,-10) circle (0pt) node[anchor=north]{$V_3$};
    
    \draw[teal, thick] (13,-6) -- (15,-6);  
    
    \filldraw[black] (13,-6) circle (2pt) node[anchor=south]{$u_2$};
    \filldraw[black] (15,-6) circle (2pt) node[anchor=south]{$u_1$};
    \filldraw[black] (17,-6) circle (2pt) node[anchor=south]{$u_0$};
    \filldraw[black] (13,-9) circle (2pt) node[anchor=north]{$w_2$};
    \filldraw[black] (15,-9) circle (2pt) node[anchor=north]{$w_1$};
    \filldraw[black] (17,-9) circle (2pt) node[anchor=north]{$w_0$};
\end{tikzpicture}

\caption{\change{A different star decomposition for the example in Figure \ref{fig:union-of-stars-example-1}. In Theorem \ref{unweighted}, choosing a different ordering for the vertices can result in a decomposition into a larger number of stars and therefore in a larger number of Ising operations. The ordering for vertices in this case is $u_0, w_0, w_1, w_2, u_2$, while the ordering for Figure \ref{fig:union-of-stars-example-1} is $u_1, u_2$. Notice that the maximum degree for any star is $2$ in this decomposition, which could possibly be useful under different error assumptions.}}
\label{fig:union-of-stars-example-2}
\end{figure}

We denote a graph $G$ as $G = (V, E, z)$, where $z \in \mathbb{R}^{|E|}$ is the weight function on edges $E$. We first claim that \rep{pulse}{operation} sequences for constructing two different edge weights on the same graph can be combined. In particular, this will imply that we can simply augment the \tsout{pulse} sequences for disjoint subgraphs (with weight 0 on non-edges) to construct the \change{target coupling} graph.

\begin{lemma} \label{composition-weighted}
    For weighted graphs $G_1 = (V, E, z_1)$ and $G_2 = (V,  E, z_2)$ with vertex set $V$ and edge set $E$ with weights $z_1, z_2 \in \change{\mathbb{R}^{|E|}}$ respectively, if $G = (V, E, z_1 + z_2)$, then $\gc(G) \le \gc(G_1) + \gc(G_2)$.
\end{lemma}

\begin{proof}
    Let $A_1, A_2$ be the adjacency matrices of $G_1, G_2$ respectively (i.e., $A_{i,j} = z_{(i,j)}$ if $(i,j)\in E$ and 0 otherwise). Then, the adjacency matrix of $G$ is $A_1 + A_2$.
     
    \medskip
    Let $A_1 = \big( P_1^TW_1P_1 \big) \odot J$ and $A_2 = \big( P_2^TW_2P_2 \big) \odot J$, where $P_1$ is a $k_1 \times n$ matrix and $P_2$ is a $k_2 \times n$ matrix for some $k_1, k_2$. We construct the matrix $P \in \{-1, 1\}^{(k_1 + k_2) \times n}$ by augmenting $P_1$ and $P_2$ as follows: for all $j \in [1, n]$,
    \[
        P_{i, j} = \begin{cases}
            (P_1)_{i, j} \;\; \text{if} \; i \le k_1 \\
            (P_2)_{i - k_1, j} \;\; \text{if} \; k_1 < i \le k_1 + k_2.
        \end{cases}
    \]
    That is, the first $k_1$ rows are the matrix $P_1$, and the next $k_2$ rows are the matrix $P_2$. We define a diagonal square matrix $W$ of size $k_1 + k_2$ as follows:
    \[
        W_{i, i} = \begin{cases}
            (W_1)_{i, i} \;\; \text{if} \; i \le k_1, \\
            (W_2)_{i - k_1, i - k_1} \;\; \text{if} \; k_1 < i \le k_1 + k_2. 
        \end{cases}
    \]
    That is, the diagonal entries in the first $k_1$ rows in $W$ are the diagonal entries in $W_1$, and those in the next $k_2$ rows are the diagonal entries in $W_2$. From a simple calculation,
    \[
        P^TWP \odot J = P_1^T W_1 P_1 \odot J + P_2^T W_2 P_2 \odot J = A_1 + A_2
    \]
    
    Choose $k_1 = \gc(G_1)$ and $k_2 = \gc(G_2)$, so that $\gc(G) \le k = \gc(G_1) + \gc(G_2)$. This proves our claim.
\end{proof}

\begin{corollary} \label{composition-unweighted}
    For unweighted graphs $G_1 = (V, E_1)$ and $G_2 = (V, E_2)$ where $E_1 \cap E_2 = \emptyset$, if $G = (V, E_1 \cup E_2)$, then $\gc(G) \le \gc(G_1) + \gc(G_2)$.
\end{corollary}

In many cases, one can do better than adding graph coupling numbers together by removing duplicate rows in \tsout{the pulse matrix} $P$. For a matrix $P$ and a row vector $r$ of $P$, let $P \setminus r$ denote the matrix $P$ with row $r$ removed.

\begin{lemma}\label{redundant-rows}
    Suppose matrices $P, W$ satisfy $\big(P^TWP\big) \odot J = A$, where $A$ is the adjacency matrix of some graph $G$. If there are rows $r, s$ in $P$ such that $r = \pm s$, then there exists matrix $W_1$ such that $\big((P \setminus s)^T W_1 (P\setminus s)\big) \odot J = A$.
\end{lemma}

\begin{proof}
    Let $r = \eta s$, where $\eta \in \{-1, 1\}$. Let $w_r, w_s$ be the diagonal entries in $W$ corresponding to rows $r, s$ respectively in $P$. To obtain matrix $W_1$ from $W$, delete the row and column containing $w_s$ and replace $w_r$ by $w_r + \eta w_s$. Then, it is easily seen that $\big((P \setminus s)^T W_1 (P \setminus s)\big) \odot J = \big(P^T W P\big) \odot J = A$.
\end{proof}

We next show that any biclique (i.e., complete bipartite graph) and a union of isolated vertices can be constructed using \rep{4 pulses}{four Ising operations}, which will give us an upper bound on the graph coupling number for arbitrary graphs. 

\begin{lemma}\label{complete-bipartite-construction}
    Given a weighted graph $G = (V, E, z)$, where $V$ is the disjoint union of $V_1, V_2, V_3$, (i.e., $V = V_1 \bigcup V_2 \bigcup V_3$ and $V_i \cap V_j = \emptyset$ for $(i\neq j)$), $E = \big\{(u, v) : u \in V_1, v \in V_2 \big\}$, and $z_e = \mu$ for all $e \in E$, (that is, $G$ is a complete bipartite graph except some isolated vertices, with all edge weights being equal), $\gc(G) \le 4$.
\end{lemma}

\begin{proof}
    We give matrices $P, W$ with $k = 4$ such that $P^TWP \odot J = A$. We let column $j$ of $P$ correspond to vertex $v_j \in V$. To define entry $P_{a, j}$ of the matrix $P$ for each $a \in \{1, 2, 3, 4\}$ and $j \in \{1, \ldots, n\}$, we define numbers $P_{a, V_1}, P_{a, V_2}, P_{a, V_3}$ and let
    \[
        P_{a, j} = \begin{cases}
            P_{a, V_1} & \text{if}\; v_j \in V_1, \\
            P_{a, V_2} & \text{if}\; v_j \in V_2, \\
            P_{a, V_3} & \text{if}\; v_j \in V_3.
        \end{cases}
    \]
    We now give matrix $W$ and specify entries $P_{a, V_l}$ for $l \in \{1, 2, 3\}, a \in \{1, 2, 3, 4\}$:
    \begin{align*}
        & \;\; \begin{matrix}
            V_1\; & V_2\; & V_3
        \end{matrix}
        \\
        P = &\begin{pmatrix}
        1 & 1 & -1\\
        1 & -1 & -1 \\
        1 & 1 & 1 \\
        1 & -1 & 1
        \end{pmatrix}, 
        & W = \begin{pmatrix}
        \frac{\mu}{4} & 0 & 0 & 0 \\
        0 & -\frac{\mu}{4} & 0 & 0 \\
        0 & 0 & \frac{\mu}{4} & 0 \\
        0 & 0 & 0 & -\frac{\mu}{4}
        \end{pmatrix}.
    \end{align*}
    Note that $\big(P^TWP \big)_{i, j} = \sum_{a = 1}^kP_{a, i}P_{a, j}W_{a, a}$. Suppose $v_i \in V_1, v_j \in V_2$. We show that the edge weight of $(v_i, v_j)$ in our construction is $\mu$, as expected: since $i \neq j$, 
    \begin{align*}
        \big(P^TWP \odot J\big)_{i, j} = \Big[(1)(1)\frac{\mu}{4} + (1)(-1)\frac{-\mu}{4} + (1)(1)\frac{\mu}{4} 
        \\
        + (1)(-1)\frac{-\mu}{4}\Big] \times 1 = \mu = A_{i, j}.
    \end{align*}
    
    A similar calculation shows that when $v_i \in V_1, v_j \in V_3$, $\big(P^TWP \odot J \big)_{i, j} = 0$ since there is no edge between $V_1, V_3$ and therefore $A_{i, j} = 0$. One can check this for every possibility of $v_i, v_j$, so that $P^TWP \odot J = A$. Therefore, $\gc(G) \le 4$.
\end{proof}

The above decomposition is crucial in constructing graphs  edge-by-edge while incorporating arbitrary edge weights, i.e., decompose the graph into a union of edges (wherein each edge is a trivial biclique), which gets us the following theorem.

\begin{theorem}\label{weighted}
   For any weighted graph $G = (V,E, z)$, $\gc(G) \le 3m+1 = O(m)$, where $m$ is the number of edges in $G$.
\end{theorem}

\begin{proof}
    Suppose we are given a weighted graph $G$ on $n$ vertices. Then, every edge can be constructed in at most $4$ steps by the above lemma: given an edge $e = (u, v) \in E$, choose $A = \{u\}, B = \{v\}, C = V - \{u, v\}$, so that $E = \big\{e\big\}$, and $\mu = z_e$. Lemma \ref{composition-weighted} then implies that $\gc(G) \le 4m$. Further note that the third row in the matrix $P$ is all ones, and this is common across the constructions for each edge. Therefore, from Lemma \ref{redundant-rows}, we can combine these rows into a single row by summing their strengths, thus giving us an upper bound of $3m+1$ on the total number of \rep{pulses}{Ising operations}.
\end{proof}

\tsout{Lemma \ref{complete-bipartite-construction} implies that the number of pulses required to construct any given graph is at most four times the minimum size of a partition of the edge set into complete bipartite graphs, which can be further reduced to $3\:\bp(G)+1$ using the observation that the all ones row is common across each biclique's construction, i.e., $\gc(G) \leq 3\:\bp(G)+1$. The biclique partition number of an arbitrary graph is however NP-hard to compute, and therefore, there is a need for approximate constructions.}

\rep{For unweighted graphs, we show that a decomposition of the}{Although for general weights, we give a bound of $3m+1$ total Ising operations, we can do much better for unweighted graphs. We show that there exists a decomposition of any graph's (unweighted)} edge set into at most $n-1$ star graphs (i.e., bicliques where one side of the partition has only one vertex), which can be used to give a linear bound on the graph coupling number. We refer to this construction as the {\it union-of-stars}.

\begin{theorem}\label{unweighted}
   For any unweighted graph $G = (V,E)$, $E$ can be partitioned into $n-1$ star graphs (i.e., $K_{1,s}$ for $s<n$) which in turn implies that the graph coupling number $\gc(G) \le 3n-2 = O(n)$.
\end{theorem}

\begin{proof}
    We first show that we can write the edge-set of $G = (V, E)$ as the disjoint union of at most $n - 1$ stars. Consider an arbitrary ordering of vertices $v_1, \hdots, v_n$. For $1 \le i \le n - 1$, define the star $S_i = \{(v_i, v_j) \in E: i < j\}$, that is, it consists of all edges $(v_i, v_j)$ where $i$ is the lower index. Each $S_i$ is a star and \tsout{that} each edge in $E$ belongs to exactly one $S_i$. Therefore, $E$ is the disjoint union of at most $n-1$ non-empty stars.

    We now show how to construct each $S_i$ in at most $4$ steps: since $G$ is unweighted, choose $\mu = 1$ in Lemma \ref{complete-bipartite-construction}. Then, each $S_i$ (with possibly other isolated vertices) can be constructed in at most $4$ steps since a star is also a biclique. This, combined with Corollary \ref{composition-unweighted}, implies that $\gc(G) \le 4(n - 1)$, which after combining rows with all ones in each construction give us the bound in the theorem, by Lemma \ref{redundant-rows}.
\end{proof}

The bound in Theorem \ref{unweighted} can be improved by finding a smaller sized partition of the edge set of the graph into bicliques \change{(using Lemma \ref{complete-bipartite-construction}), giving a bound of $\gc(G) \leq 3\:\bp(G)+1$}. However, \tsout{it is known that for complete graphs on $n$ vertices,} the biclique partition is NP-hard to compute. \change{Moreover, it is known that for complete graphs on $n$ vertices, the biclique partition number $\bp(K_n) = n - 1$} \cite{graham_pollak_loop_switching}, so that the upper bound on $\gc(G)$ still remains $3n-2$ (using union-of-stars construction). \rep{Moreover}{Note that} in any construction, one can further \rep{get a reduction in}{reduce} the $L_0$ norm by removing repeating rows in $P$, using Lemma \ref{redundant-rows}, although the amount of such a reduction can be instance dependent.

\subsection{A Lower Bound}
We next discuss a lower bound on the graph coupling number of any arbitrary graph with distinct edge weights:

\begin{lemma}
   For each $n$, there is a weighted graph $G$ with $n$ vertices such that $\gc(G) = \Omega(\log n)$.
\end{lemma}

\begin{proof}
   For any $k$ note that $\Big(P^TWP \odot J\Big)_{i, j} = \displaystyle \sum_{a = 1}^k P_{a, i} P_{a, j} W_{a,a}$ if $i \neq j$ and it is equal to $0$ otherwise. That is, each non-diagonal entry of $P^TWP \odot J$ is a linear combination of \rep{pulse}{Ising operation} strengths $W_{a,a}$ with coefficients either $1$ or $-1$ (since $P \in \{-1, 1\}^{k \times n}$). The set of all such linear combinations has cardinality at most $2^k$. Consider any complete graph on $n$ vertices with distinct edge weights. Then, there are at least $\frac{n(n - 1)}{2}$ distinct entries in $A$. When $\frac{n(n - 1)}{2} > 2^k + 1$, by our previous observation, $A \neq P^TWP \odot J$ for \emph{any} diagonal matrix $W \in \mathbb{R}^{k \times k}$ and matrix $P \in \{\pm 1\}^{k \times n}$. That is, for $P^TWP \odot J = A$ to hold, we need $k \ge \log_2 \Big(\frac{n(n - 1)}{2} - 1\Big)$, which implies that $\gc(G) = \Omega(\log n)$.
\end{proof}

\subsection{Optimal Solutions Through Brute Force}\label{sec:brute}
We describe a mixed-integer program (MIP) that can be solved to optimality. In the preceding sections, we have given polynomial constructions for generating both weighted and unweighted coupling graphs, which gave upper bounds for the graph coupling number. It is important to note that MIPs are not generally solvable in polynomial time; therefore, this is not an efficient method for finding \rep{pulse}{operation} sequences. Instead, by solving the MIP to optimality, we identify the optimality gap and quantify the potential gains if a more efficient construction can be found.

To convert the graph coupling problem into a MIP, we construct a {\it complete} \rep{pulse}{operation} matrix $P$ by enumerating all possible rows with elements $\pm1$. In other words, we must consider all possible choices of bit flips to prove a given sequence is optimal. Since negating a row in $P$ does not change $P^TWP$, there are $2^{n-1}$ unique rows of $P$ to consider. Thus, finding the optimal sequence to produce a coupling graph with adjacency matrix $A$
reduces to finding a strength matrix $W$ such that $P^TWP \odot J = A$. The only variable is the diagonal strength matrix $W$, and the objective is to minimize its $L_0$ or $L_1$ norm. A complete description of the MIP is included in \refapp{app:MIP}.

We note that in practice, we have observed that sub-sampling the complete $P$ matrix and running the MIP for a fixed amount of time often produces a tractable \textit{good} solution (i.e., better running time and memory requirements), however in this case, we cannot obtain provable bounds to optimality. In the next section, we compare optimal solutions to the constructions for small graphs of up to 8 vertices (optimal up to the choice of the parameter $M$, which was set to the sum of the edge weights in the simulations). The brute force search to optimality is time intensive with some graphs taking hours to complete or even longer than 24 hours \footnote{The MIP was solved using Python 3.5 and Gurobi 9 on a \tsout{high performance computing (HPC)}{computing} cluster of machines with an Intel Xeon E5-2670 2.3-GHz CPU and 128GB of main memory.}.

\section{Numerical Experiments\tsout{: Optimal Results}}
\change{In order to test the utility of the union-of-stars construction, we conduct two sorts of numerical experiments. In \refsec{sec:results}, we compare the required resources to the derived upper bounds and the optimal controls obtained by brute force. To demonstrate the utility of running Max-Cut QAOA with the union-of-stars construction of the cost Hamiltonian, in \refsec{sec:noisy} we compare the simulated performance of our construction and the standard CNOT construction in the presence of noise.}

\subsection{\change{Optimal Results}}\label{sec:results}
We compare the union-of-stars construction with optimal \footnote{subject to the $M$-constraint described in the previous section and in \refapp{app:MIP}} sequences of uniform $\HIsing$ operations, where $J$ is the adjacency matrix of the graph $K_n$.
First, we generated 96 Erd{\~o}s-R{\'e}nyi random unweighted graphs with 7 vertices, where each edge has a probability $p$ to exist, with 4 random graphs for each $p \in 0.04 \times \{1,2,...,24\}$. For each graph, we determined the optimal solution using the brute force MIP optimizing for either the minimal $L_0$ or $L_1$ norm.
While the union-of-stars construction has upper bound $3n-2$ for the number of \rep{$\HIsing$ pulses}{Ising operations} ($L_0$ norm), in practice many graphs can be constructed with fewer than $n-1$ stars. We chose to start with the largest star (highest degree vertex), then add the next largest star, and so on until the target graph was realized. In addition, we check each $P$ matrix for identical rows and combine them by summing the corresponding strengths.

\reffig{fig:random-unweighted} compares the union-of-stars solution with the upper bound and the brute force optimal solution. We observe that the union-of-stars typically performs within a factor of two of optimality in the $L_0$ norm, while the $L_1$ norm is closer to a factor of 3 worse.

\begin{figure}
    \subfloat[unweighted graphs, $L_0$ norm ]{{\includegraphics[width=0.47\textwidth]{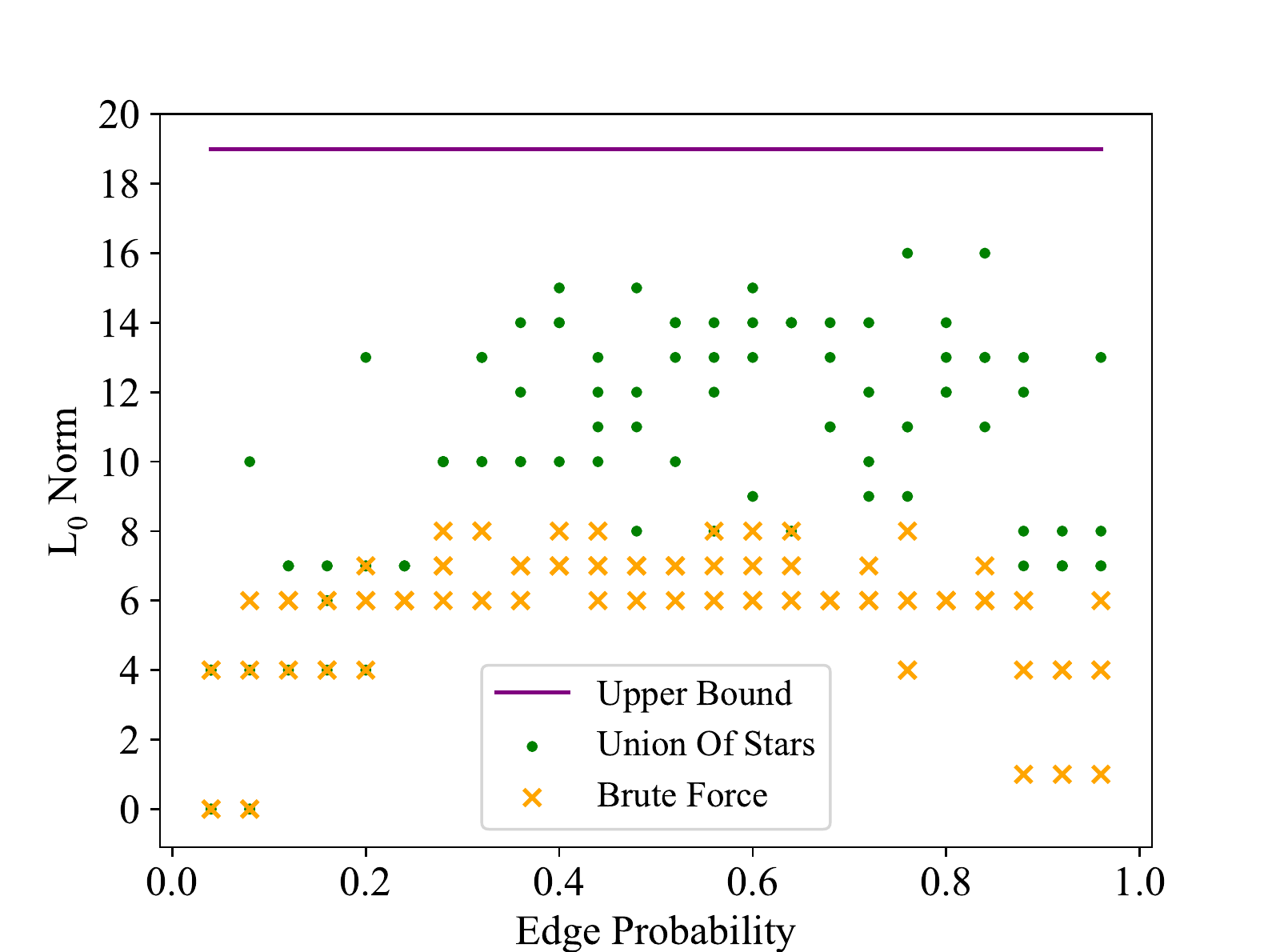} }}
    \hfill
    \subfloat[unweighted graphs, $L_1$  norm]{{\includegraphics[width=0.47\textwidth]{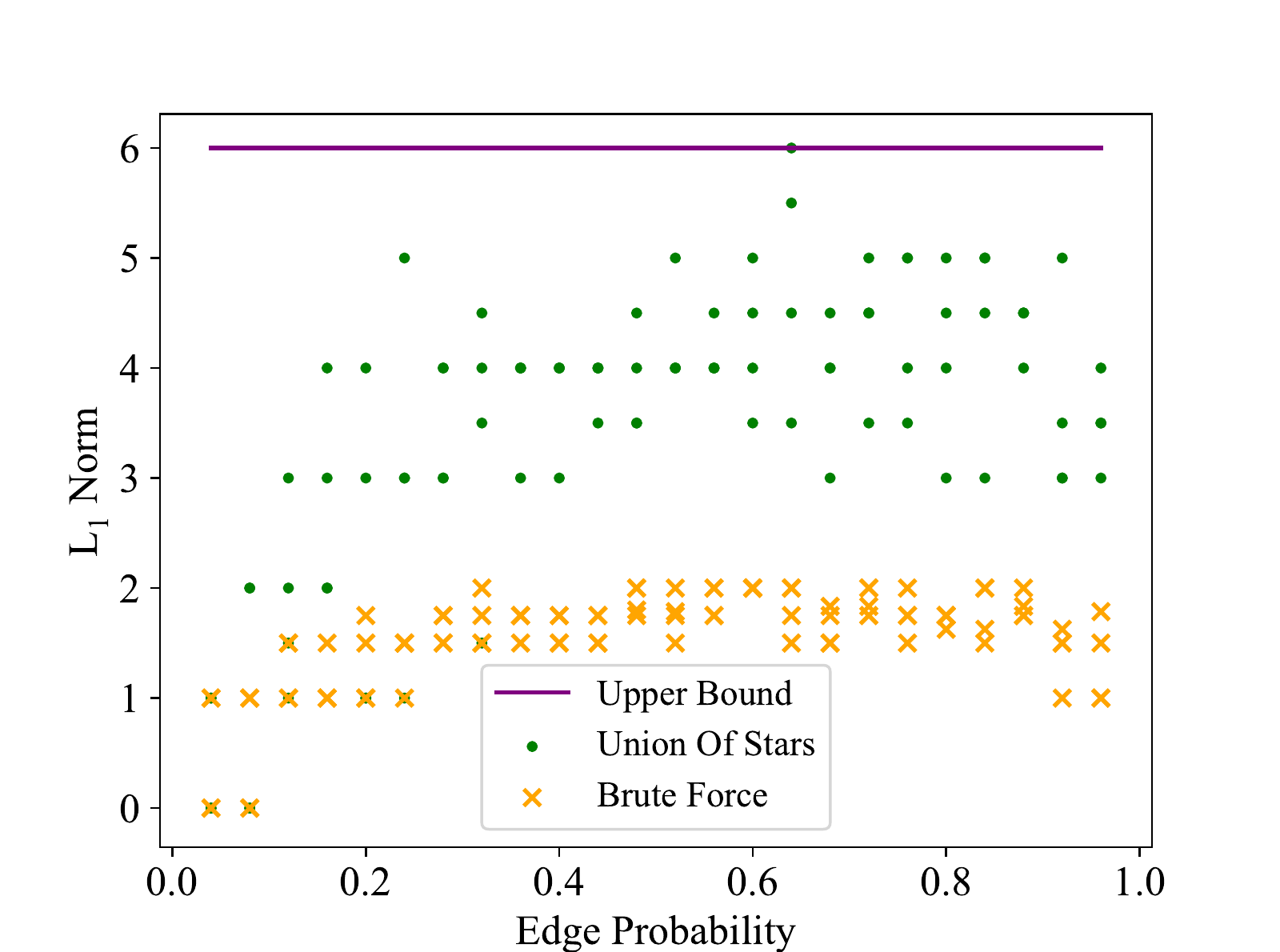} }}
    \caption{A comparison of the union-of-stars construction to the optimal sequence found with the brute force MIP for random unweighted graphs with 7 vertices. The upper bound proved with the union-of-stars is $3n-2$ for the $L_0$ norm (total number of \rep{$\HIsing$ pulses}{Ising operations}) and $n-1$ for $L_1$ (sum of the absolute magnitude of strengths).}
    \label{fig:random-unweighted}
\end{figure}

Next, we tested random weighted graphs by randomly assigning edge weights sampled uniformly from $\{1,2,3\}$ to Erd{\~o}s-R{\'e}nyi graphs with seven vertices and similar edge inclusion probabilities. As shown in \reffig{fig:random-weighted}, the union-of-stars method again performs well.
For comparison, the \textit{expected} upper bounds of $3\bar{m}+1$ for the $L_0$ norm and $3\bar{m}$ for the $L_1$ norm are plotted, where $\bar{m} = pn(n-1)/2$ is the expected number of edges for edge inclusion probability $p$.
Notably, solving the MIP to optimality for the $L_0$ norm takes significantly longer compared to the unweighted graphs. Some instances in \reffig{fig:random-weighted}(a) timed out (24 hrs) before reaching optimality; in this case, the sub-optimal values are plotted instead.
In \reffig{fig:random-weighted}(b), we observe that when we optimize instead for $L_1$ norm on weighted graphs, the optimal solution is much better than the union-of-stars construction for graphs with many edges. 

\begin{figure}
    \subfloat[weighted graphs, $L_0$ norm ]{{\includegraphics[width=0.47\textwidth]{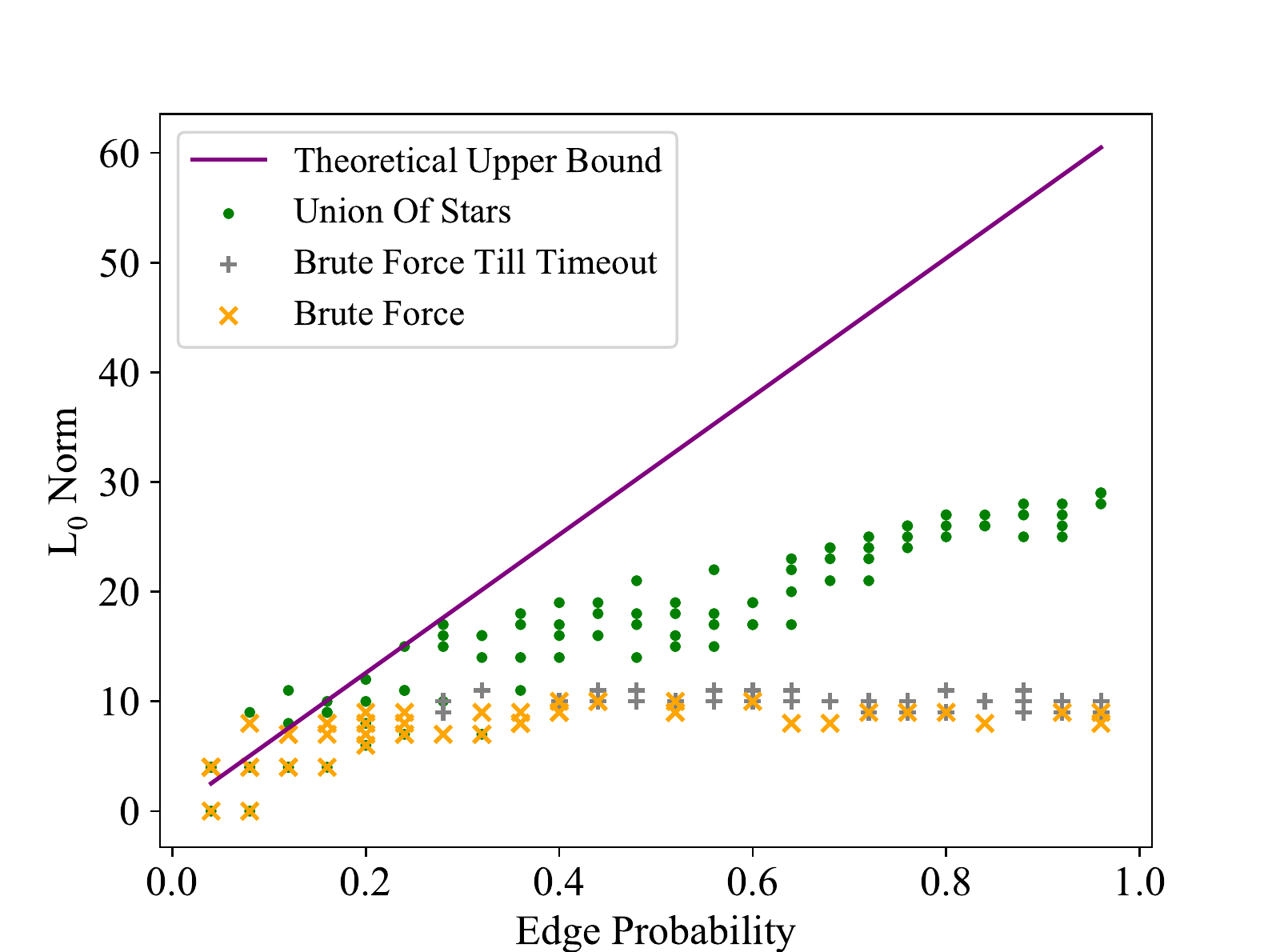} }}
    \hfill
    \subfloat[weighted graphs, $L_1$  norm]{{\includegraphics[width=0.47\textwidth]{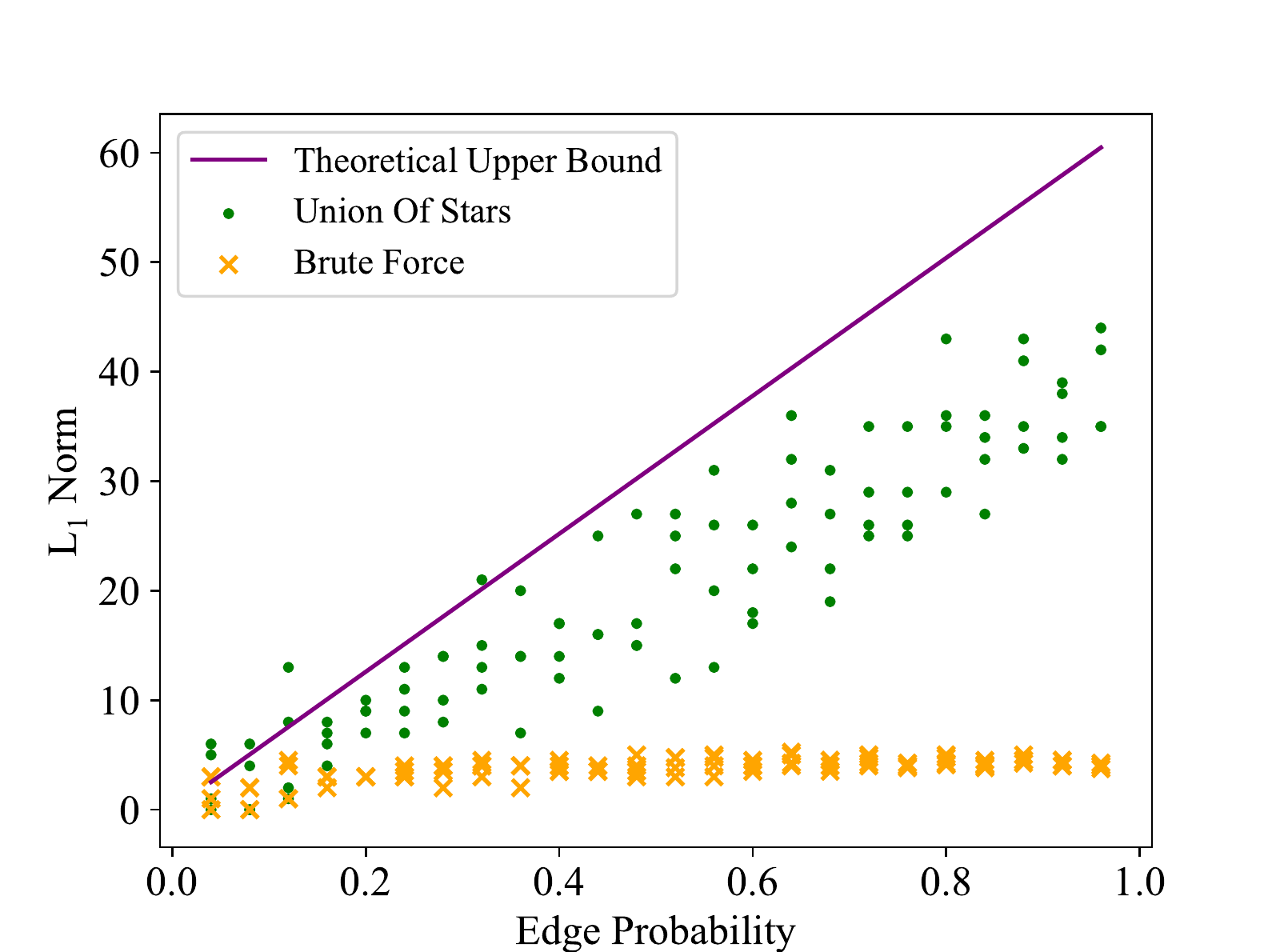} }}
    \caption{A comparison of union-of-stars against optimality for random weighted graphs on 7 vertices. The upper bound now depends on the \textit{expected} number of edges as described in the main text. That some constructions exceed this expectation reflects the non-zero probability of having more edges than $\bar{m}$.}
    \label{fig:random-weighted}
\end{figure}

Finally, we used the MIP to find the worst case graph coupling number ($L_0$ norm) for all non-isomorphic unweighted graphs with up to 8 vertices, using the enumeration of \citet{graphenumeration}. These results are given in \reffig{fig:fig3}.
\begin{figure}
    \includegraphics[width=0.47\textwidth]{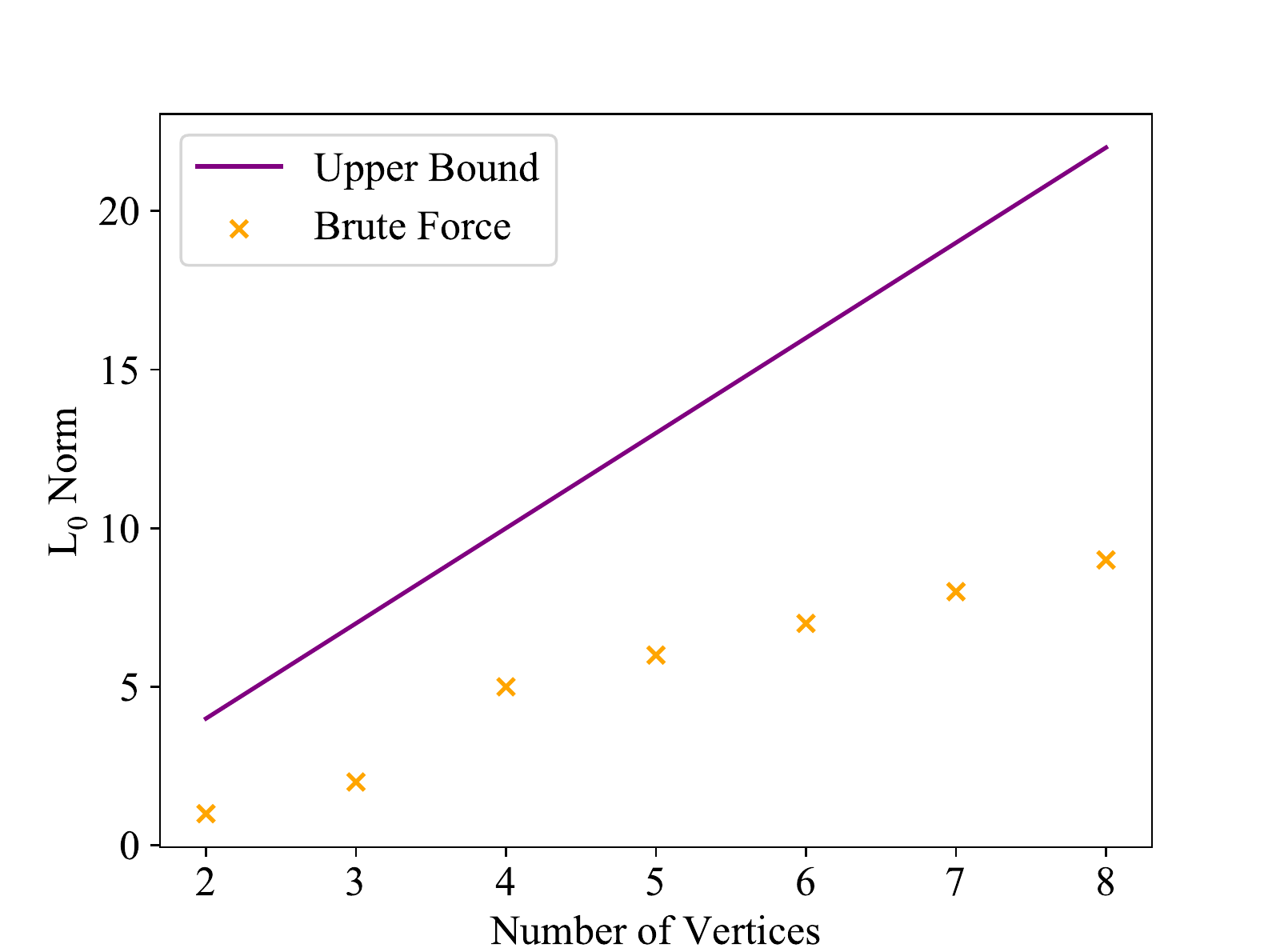}
    \caption{Largest graph coupling number observed in the optimal generation of unweighted graphs of up to 8 vertices compared to the $3n-2$ upper bound from the union-of-stars construction.}
    \label{fig:fig3}
\end{figure}
While the linear pattern we see between the worst case graph coupling number and the number of vertices does not necessarily hold for larger graphs, it leads us to suppose that the number of \rep{$\HIsing$ pulses}{Ising operations} required to construct any unweighted graph on $n$ vertices with uniform ${\HIsing}$ and single-qubit bit flips may be bounded from above by $n+1$ (i.e., graph coupling number is linear), and, therefore, the union-of-stars method is order-optimal for small graphs.

\subsection{\change{Noisy QAOA}}\label{sec:noisy}

\change{In the standard QAOA \cite{farhi_quantum_2014}, a cost and a mixing operator are applied in an alternating fashion in order to drive the solution to an extremal state of the cost operator. 
The quantum hardware is programmed to prepare the state}
\begin{equation}
    \change{\ket{\vec{\gamma},\vec{\beta}} = \prod_{k=1}^p e^{-i \beta_k B} e^{-i \gamma_k C^{\prime}} \ket{+}^{\otimes n}} \,,
\end{equation}
\change{where the mixing operator is $B=\sum_{i=1}^n \sigma_i^x$ and the Max-Cut cost operator $C^{\prime}$ is related to the coupling operator $C$ (\refeq{eq:coupling})}:
\begin{equation}
    \label{eq:Max-Cut-cost}
    \change{C^{\prime} = \sum_{i=1}^{n-1} \sum_{j=i+1}^{n} a_{i,j} \frac{1-\sigma_i^z \sigma_j^z}{2} =  - \frac{C}{2} + \sum_{i=1}^{n-1} \sum_{j=i+1}^{n} \frac{a_{i,j}}{2}}\,.
\end{equation}
\change{After sampling the quantum state, the estimated expectation value of the cost, $\braket{C^{\prime}}=\braket{\vec{\gamma},\vec{\beta}|C^{\prime}|\vec{\gamma},\vec{\beta}}$, is computed, and the $2p$ parameters $\vec{\gamma},\vec{\beta}$ are classically optimized through repeated calls to the quantum hardware, increasing the likelihood of observing a Max-Cut solution.}

\change{The ZZ-coupling terms found in $C^{\prime}$ can be implemented in a quantum circuit in various ways, depending on the physical quantum architecture and its natural, native gates. For digital quantum computers, the cost unitary is usually programmed with two CNOT gates and a parameterized Z-rotation (\reffig{fig:standardZZ}) for each edge in the problem graph. For a graph with $m$ edges, this requires at least $2m$ CNOT gates, and more if SWAPs are required to connect distant qubits. Here we compare and contrast this ``standard'' compilation to the union-of-stars method presented previously, including quantum error channels on all gates.}

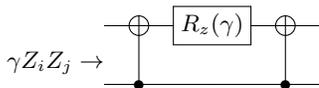
\begin{figure}[!tb]
           \[ \Qcircuit @C=1em @R=0.75em {
&	&	&		&		&		&	\targ	&	\gate{R_z(\gamma)}	&	\targ	&	\qw	\\
&		&		&	\gamma Z_iZ_j\rightarrow	&		&		&		&		&		&		&		&		\\
&		&	&		&		&		&	\ctrl{-2}	&	\qw	&	 	\ctrl{-2}	&	\qw
} \]
\caption{\change{Decomposition of the $Z_iZ_j$ terms into CNOT gates. Index $i,j$ directly map from graph edges to arbitrary qubits, requiring either arbitrary connectivity or inserting  swap gates to accommodate architectures with limited connectivity. $R_{z}(\gamma )=\exp (-i\gamma \sigma _{z}/2)$.}}
\label{fig:standardZZ}
\end{figure}

\change{Since the two methods scale differently in terms of their gate count, we expect different performance when noise is considered. To test this, we investigate both compilation models with four example graphs when using a composite quantum noise model simulated in Qiskit~\cite{Qiskit}.
Specifically, we consider depolarizing noise channels that act on all CNOT gates and global M{\o}lmer-S{\o}renson or Ising operations described in the rows of $P$. For CNOT operations these are necessarily two-qubit depolarizing noise channels, but for the Ising operations used in the union-of stars method, we consider $N$-qubit depolarizing noise channels---a worst case assumption. We consider these as ``major'' sources of noise as these operations are likely more prone to error than single qubit operations. Additionally, we also consider ``minor'' noise sources, in which all single qubit gates also experience single qubit depolarizing and phase noise. All minor errors are fixed to occur $10\%$ as frequently as the major errors. All qubits also experience measurement noise, which is also grouped as a minor source of error~\cite{Nielsen2011}.}
\begin{figure*}
    \centering
    \subfloat[]{\includegraphics[width=0.49\textwidth]{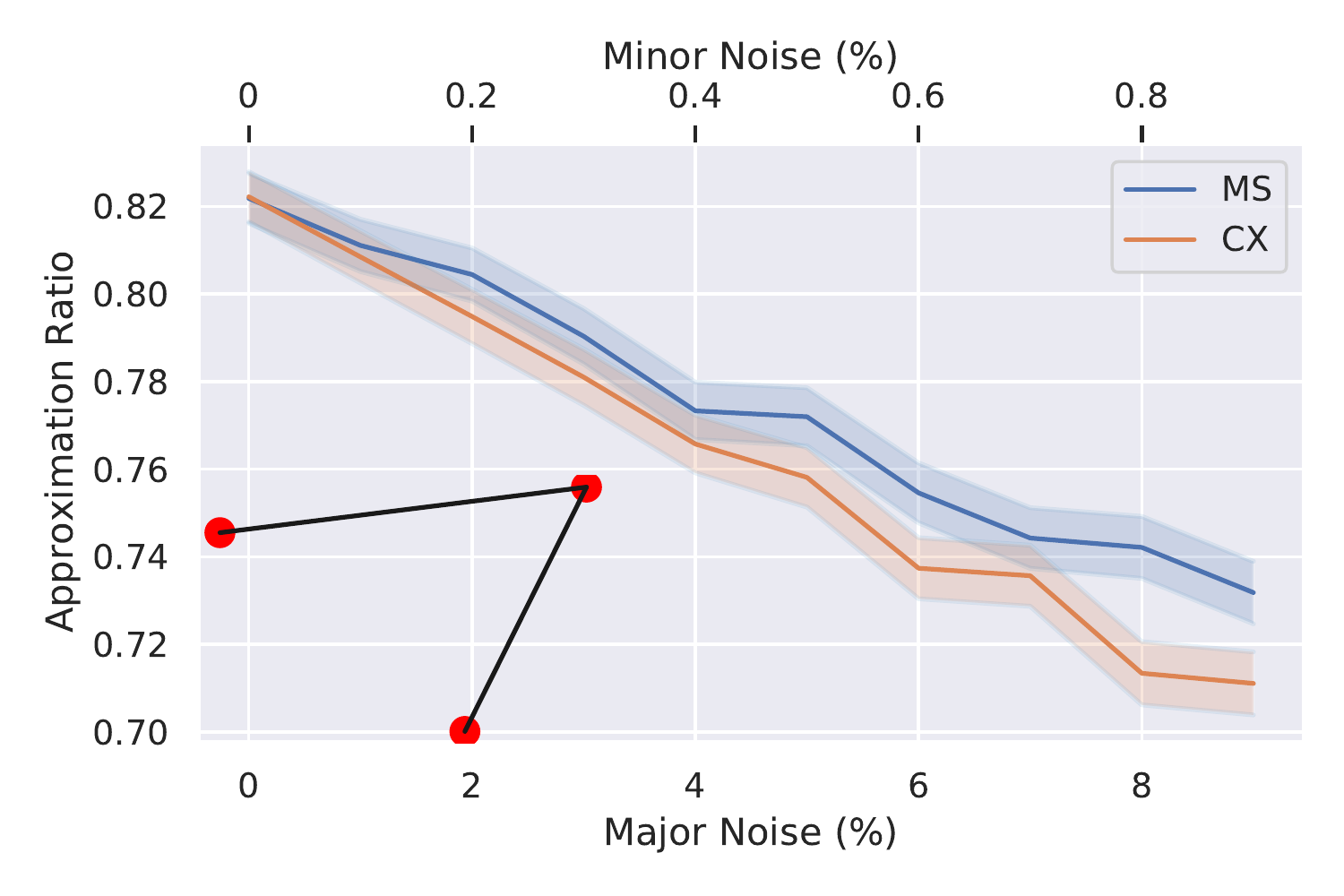}}
    \subfloat[]{\includegraphics[width=0.49\textwidth]{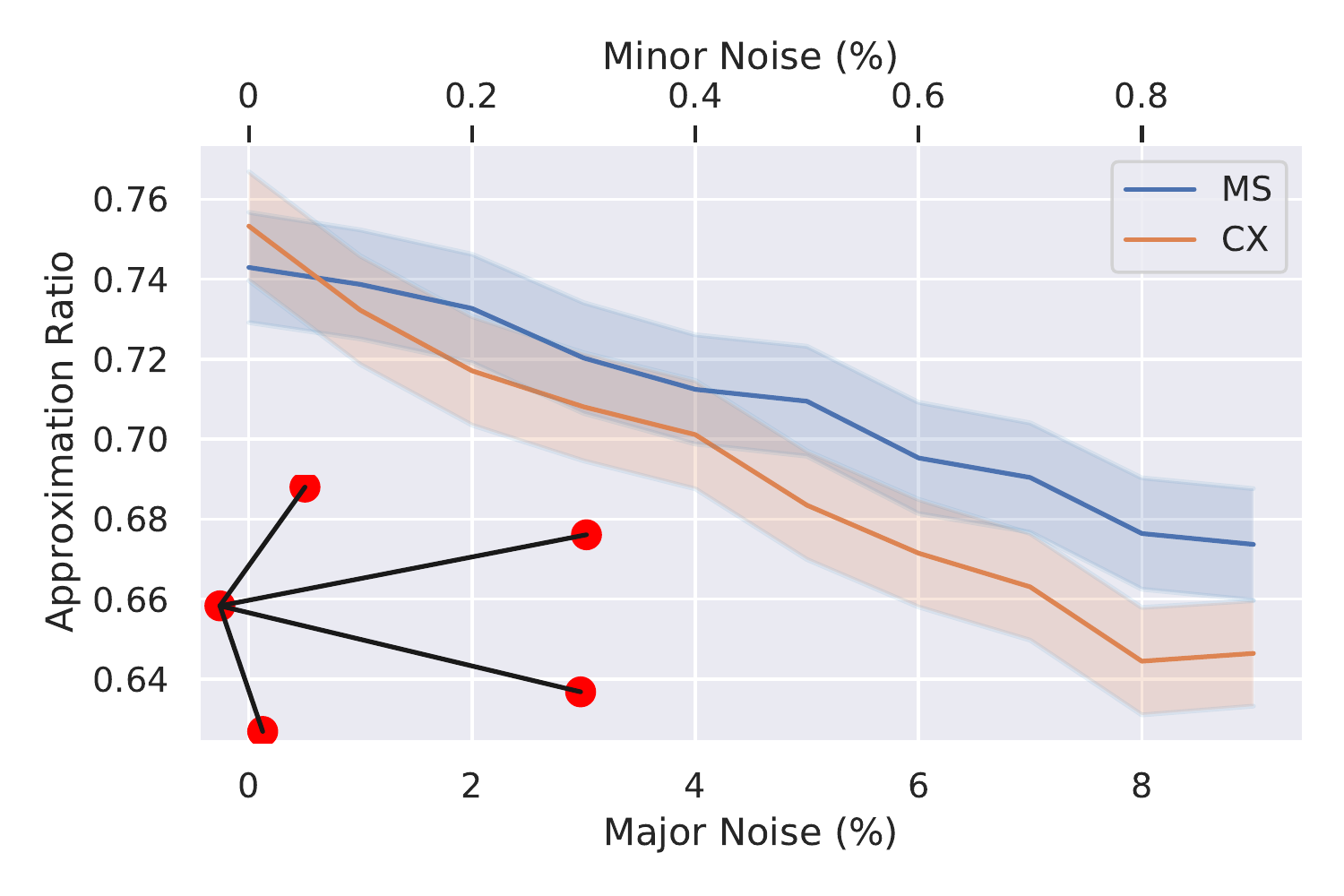}}
    \hfill
    \subfloat[]{\includegraphics[width=0.49\textwidth]{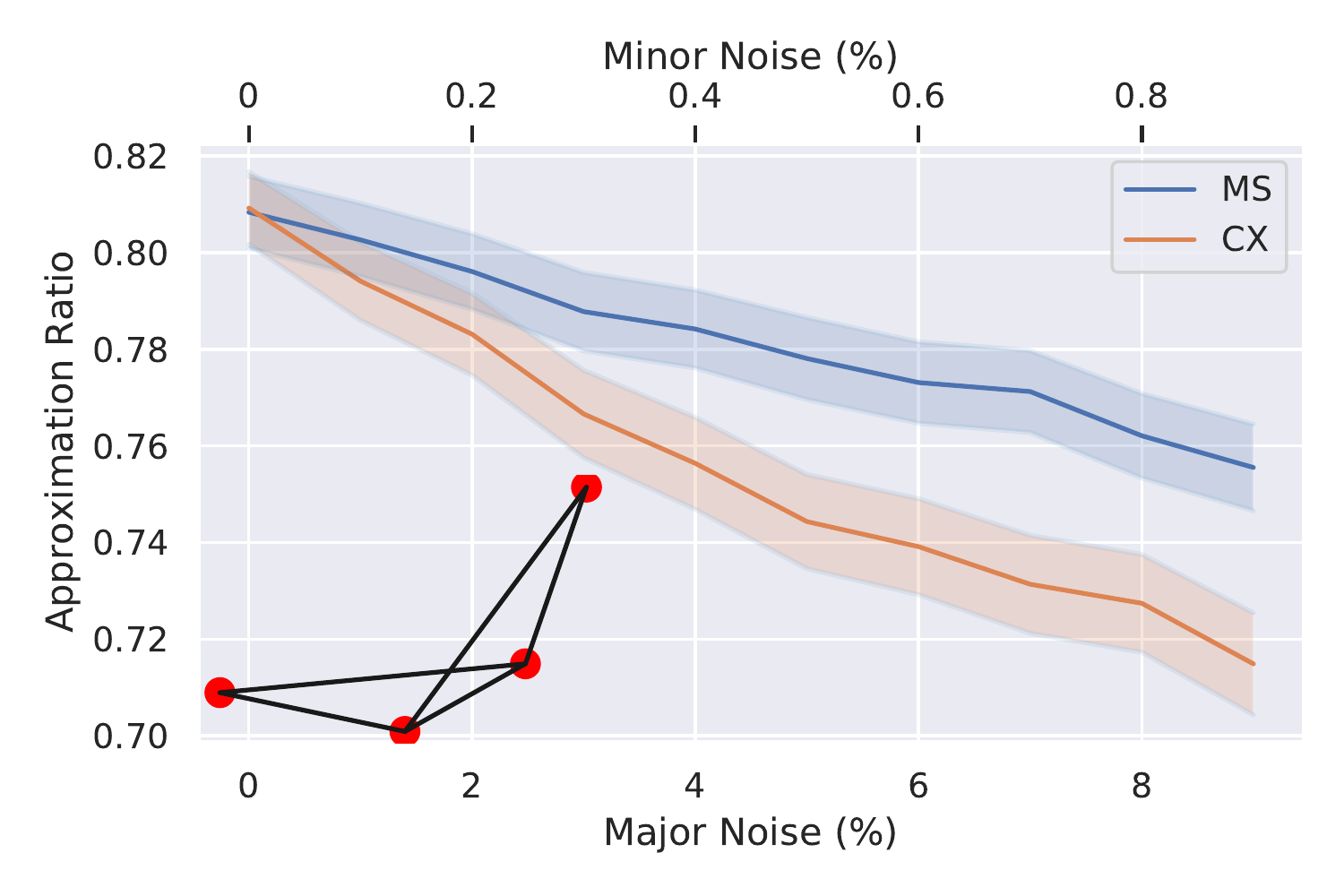}}
    \subfloat[]{\includegraphics[width=0.49\textwidth]{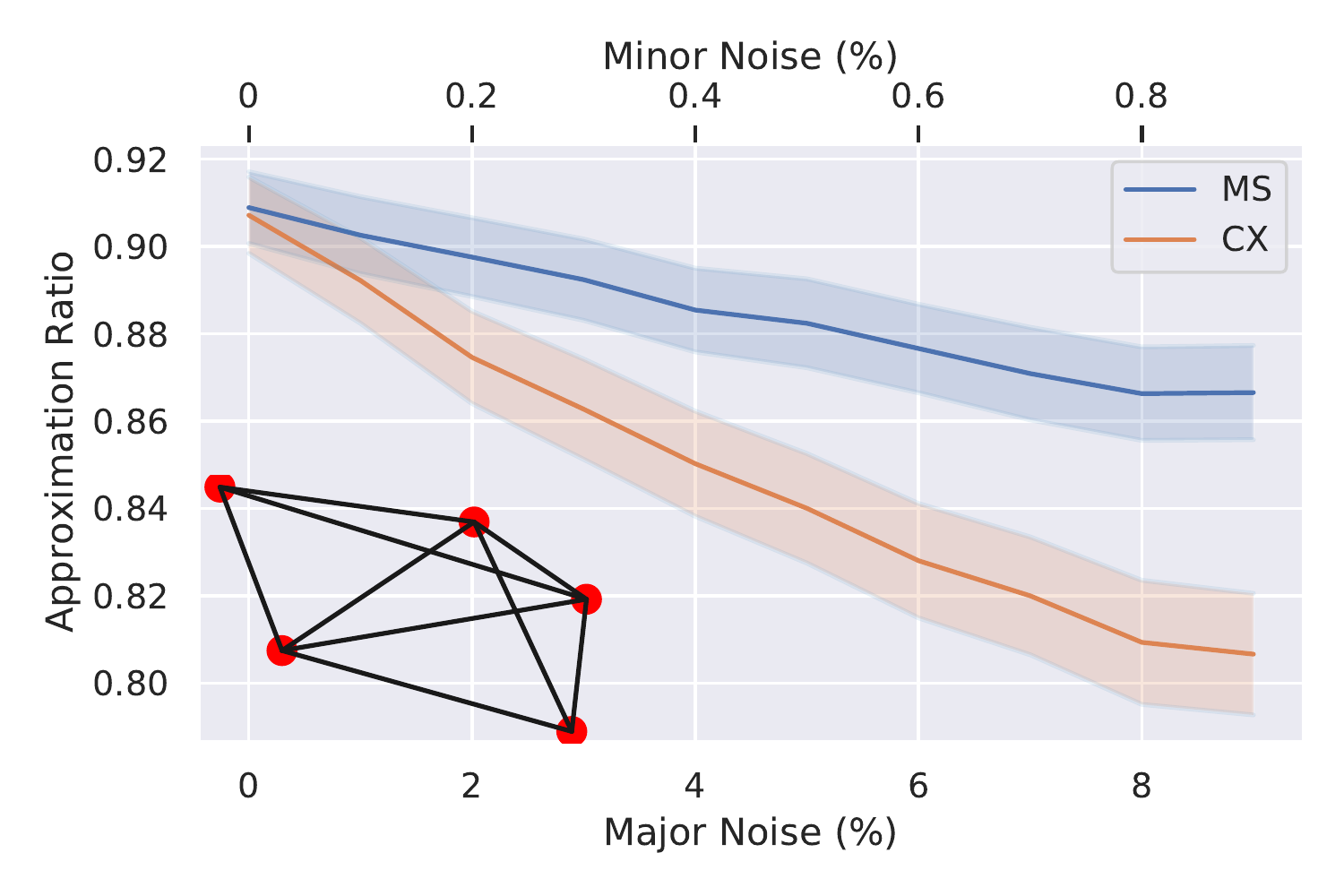}}
    \caption{\change{Approximation ratio when simulating the noisy $p=1$ QAOA, for two different compilation methods. ``CX'' compilation refers to compiling in terms of CNOT gates while ``MS'' refers to compiling following the union-of-stars method and building in terms of global M{\o}lmer-S{\o}renson interactions. Major and minor noise sources are detailed in the text and have a fixed relationship as indicated on the horizontal axis. Each subfigure shows an inset displaying the example unweighted graph that is considered for both compilation schemes.}}
    \label{fig:Compliation_Noise}
\end{figure*}

\change{In Fig.~\ref{fig:Compliation_Noise} we show {the results of} noisy simulations of Max-Cut QAOA for four small graphs at $p=1$ and find that, due to the reduced required quantum resources in the union-of-stars method, we see reduced error in the resulting approximation ratio, $\braket{C^{\prime}}/C^{\prime}_{\mathrm{max}}$. \reffig{fig:Compliation_Noise}(a) shows similar performance between the two compilation methods, while \reffig{fig:Compliation_Noise}(d) shows significantly better performance for the ``MS" compilation built from the union-of-stars method due to its fewer quantum operations. This indicates that there are classes of graphs, including dense unweighted graphs, that are more favorable to construct with the union-of-stars method.}

\subsection{\change{Resource Estimates}}\label{sec:runtime}
\change{Having shown an efficient construction of arbitrary coupling graphs through repeated applications of the Ising interaction interleaved with individual bit flips, we now estimate how long such operation sequences will take to perform. We estimate operation times based on our experience with the hyperfine qubit in \Yb \cite{Herold2016, Fallek2016, Meier2019}, however, these gate durations are similar to most trapped ion quantum computing hardware within an order of magnitude.}

\change{We assume a single-qubit bit flip takes time $T_{\pi}=5~\mu$s and a two-qubit MS gate $T_{\mathrm{MS}}=100~\mu$s. For the center of mass mode of an $n$-ion chain in a harmonic trap, $b_{j,0}=1/\sqrt{n}$, and therefore the native coupling $J_{i,j}$ scales as $1/n$ for fixed intensity and detuning, following \refeq{eq:MS-Jij}. Therefore, we estimate the time for an Ising operation is $T_{\mathrm{Ising}}=n (50~\mu$s) \footnote{For very large ion crystals, this time can be reduced by at least an order of magnitude by increasing the gate laser intensity and reducing the mode detuning.}. In the union-of-stars construction, each Ising operation is assigned a strength, and the $L_1$ norm gives the sum of the absolute magnitude of these strengths. The number of rounds of bit flips is one more than the number of Ising operations (rows of $P$), or $L_0+1$. Assuming the bit flips are performed in parallel, then the total time $\tau_G$ to implement a coupling graph is}
\begin{equation}
    \change{\tau_G = (L_0 + 1) T_{\pi} + L_1 T_{\mathrm{Ising}}}
\end{equation}
\change{For unweighted graphs, our upper bounds give $\tau_G < (3n+1) T_{\pi} + (n-1) T_{\mathrm{Ising}}$. Therefore, implementing each Max-Cut QAOA cost layer of an arbitrary, unweighted 10-node graph will take no more than 5~ms, while a 100-node graph would take at worst 500~ms, which could be further reduced by increasing laser intensity. For comparison, the $T_1$ time of the hyperfine ground state qubit in \Yb is effectively infinite and $T_2$ exceeding 10 minutes has been demonstrated \cite{Wang2017}. Implementation time is proportional to the number of QAOA layers $p$, and given that our bound is not tight, this is a pessimistic estimate. In \reffig{fig:random-unweighted}, we observe that total Ising and bit flip operation counts are often two to five times lower than this bound. When combined with further reduction of the Ising operation time, Max-Cut QAOA with hundreds of ions is feasible.}

\section{Concluding Remarks}
We have provided a method to construct arbitrary coupling operations on quantum spin-systems using only global ${\HIsing}$ operations and single qubit bit flips. With this method, the number of \rep{$\HIsing$ pulses}{Ising operations} necessary scales linearly in the number of qubits for unweighted graphs, and linearly in the number of edges for weighted graphs. An interesting consequence of this trend is that the use of global, infinite-range entangling operations may be more efficient than the use of individual two-qubit gates when applied to algorithms involving dense, unweighted coupling operations (such as QAOA). This is because the number of two-qubit gates needed to construct a coupling operation scales in the number of edges rather than vertices. Although some two-qubit gates can be performed in parallel on certain hardware, sparse hardware connectivity graphs may also require many SWAP operations.
\change{Recent work, which cites a preliminary version of this article, suggests that non-standard QAOA implementations, like the one presented here, are necessary for scaling to large Max-Cut graphs \cite{lotshaw_scaling_2022}.}

While our union-of-stars construction performs well,
we have not provided an efficient, scalable algorithm to find the \rep{optimal pulse}{\emph{optimal} operation} sequence for any arbitrary graph. We speculate that this problem is NP-hard, but \tsout{as we have been approaching it from a numerical perspective,} we have not rigorously proven NP-hardness. Therefore, there is scope for further research to answer several important questions. For example, can we solve the problem quickly and optimally for certain types of graphs? Can we establish a tighter upper bound on construction? And is this problem \rep{even NP-hard}{NP-hard in general}?

It is interesting to consider whether the complexity of compilation will ever affect what we define as the complexity of a quantum algorithm. For example, if \rep{if we apply our compilation technique to an implementation of QAOA for the Max-Cut problem}{finding an efficient compilation requires an exhaustive search}, are we just replacing \rep{an}{one} NP-hard problem with \rep{an NP-hard compilation problem}{another}? \rep{One reason why this may not be a significant concern is that we can solve the compilation problem polynomially (using union-of-stars}{For Max-Cut QAOA, our union-of-stars construction provides feasible compilations with polynomial complexity, but} without the guarantee of optimality. \tsout{and still find a feasible implementation of the target quantum circuit.} This optimality gap will only ever be a serious concern if we find that our distance to optimality in graph construction is directly correlated to our distance to optimality of our target algorithm (e.g. QAOA) due to a significant increase \change{in} noisy operation count. For now, particularly for unweighted graphs, the \tsout{sub-optimal} union-of-stars method we have presented serves as an effective solution to construct arbitrary coupling operations.

\begin{acknowledgements}
This material is based upon work supported by the Defense Advanced Research Projects Agency (DARPA) under Contract No. HR001120C0046. The authors are also grateful to Hassan Mortagy for his help in running the experiments on the \tsout{high performance} computing cluster at Georgia Tech.
\end{acknowledgements}

\appendix
\section{Example\label{app:example}}
As a simple example, we will show the construction of a target coupling graph on 3 qubits using only global ${\HIsing}$ (characterized by $J$) and singe-qubit bit flips. We will assume that our global ${\HIsing}$ is perfectly uniform, so $J$ is the adjacency matrix of $K_3$, and we will consider the construction of the simple coupling graph with a unit-weight edge between qubits 1 and 2, and a unit-weight edge between qubits 2 and 3. Therefore, we have the following $A$ and $J$,
\begin{align}
A = 
\begin{pmatrix}
0 & 1 & 0\\
1 & 0 & 1\\
0 & 1 & 0
\end{pmatrix},
J = 
\begin{pmatrix}
0 & 1 & 1\\
1 & 0 & 1\\
1 & 1 & 0
\end{pmatrix}
\end{align}
Now, we must find some diagonal $W \in \mathbb{R}^{k\times k}$ and $P \in \mathbb{R}^{k \times 3}$ so that $A = P^T W P \odot J$ is satisfied. A solution $P$ and $W$ that happens to be optimal in graph coupling number ($k$) can be found via the union-of-stars method and they are,
\begin{align}
P = 
\begin{pmatrix}
1 & 1 & 1\\
1 & -1 & 1
\end{pmatrix},
W =
\begin{pmatrix}
0.5 & 0 \\
0 & -0.5
\end{pmatrix}
\end{align}
To translate this to an experimental implementation of ${\HIsing}$ and $X_j$ \rep{pulses}{operations}, we iterate through each diagonal element of $W$ and each corresponding row of $P$. So first, because $W_{1,1} = 0.5$ we apply a global ${\HIsing}$  of strength 0.5, and because $P_1 = [1,1,1]$, we apply no $X_j$ operations. The coupling operator is given by
\begin{equation}
     C_1 = 0.5\sigma_1^z \sigma_2^z + 0.5\sigma_2^z \sigma_3^z + 0.5\sigma_1^z \sigma_3^z\,.
\end{equation}
This corresponds to a complete coupling graph $K_3$ with edge-weights of 0.5. In the second step, we apply a global ${\HIsing}$ of strength $-0.5$. We achieve the negation by switching the sign of our detuning from the center of mass mode (see \refeq{eq:constantJ}). Because $P_2 = [1,-1,1]$, we also apply $X_2$ before and after the $\HIsing$. These bit flip operations invert the basis of qubit 2 and negate coupling operation, i.e. $\sigma_x \sigma_z \sigma_x = -\sigma_z$. This coupling operator is
\begin{align}
    C_2 &= (i \sigma_2^x)(-0.5\sigma_1^z \sigma_2^z - 0.5\sigma_2^z \sigma_3^z - 0.5\sigma_1^z \sigma_3^z)(-i \sigma_2^x) \nonumber\\
    &= 0.5\sigma_1^z \sigma_2^z + 0.5\sigma_2^z \sigma_3^z - 0.5\sigma_1^z \sigma_3^z \,.
\end{align}
The overall coupling operation is
\begin{equation}
    C = C_1 + C_2 = \sigma_1^z \sigma_2^z + \sigma_2^z \sigma_3^z \,.
\end{equation}
This corresponds to the desired coupling graph, a unit-weight $ZZ$-spin coupling between qubits 1 and 2 as well as between qubits 2 and 3.

\section{MIP Formulation of Brute Force $L_0$ minimization\label{app:MIP}}

As discussed in \refsec{sec:brute}, we can use a mixed integer program to find the optimal \rep{pulse}{operation} sequence for small graphs that minimizes the $L_0$ norm of the strength matrix $W$. Recall that our problem is:
$$\min \|W\|_0: \{P^TWP \odot J = A\}, $$
where $W$ is the diagonal strength matrix, $P$ is a matrix with $\{\pm1\}$ entries, $A$ is the weighted adjacency matrix of the graph to be constructed and $J$ is the coupling matrix given by $J_{i,j}=1$ for $i\neq j$ and 0 otherwise. One can formulate a mixed integer program to track the $L_0$ norm of $W$, given all possible \rep{pulse sequences as a pulse matrix}{bit flip combinations encoded in} $P \in \{\pm1\}^{k \times n}$ ($k = 2^{n-1}$ for an $n$ node graph). Let $b_i$ denote a binary variable that is equal to 1 whenever $W_{i,i}$ is non-zero, and 0 otherwise. Then, the MIP formulation is simply:
\begin{align}
    \min &\sum_{i} b_i\\
    & W_{i,i} \leq b_i M, ~~ \text{for } i \in [k],\\
    & W_{i,i} \geq - b_i M, ~~ \text{for } i \in [k],\\
    & (P^TWP)_{i,j} = A_{i,j} ~~ \text{for } i\neq j, \\
    &b_{i} \in \{0,1\}, \text{for } i\in [k], 
\end{align}
where $M$ is an upper bound on the absolute value of the strength matrix. We show that for large enough and finite $M$ (set equal to the upper bounds on $\|W_{\star}\|_\infty$ in the following theorem), the above formulation is equivalent to minimizing the $L_0$ norm of the strength matrix $W$.

\begin{theorem}\label{thm:upperbound}
    Given a graph $G = (V, E, z)$, let $W_\star$ be the optimal strength matrix for the graph coupling problem for $G$, and let $r = \gc(G) = \|W_\star \|_0$. Also let $Z = \sum_{e \in E} |z(e)|$. Then, for $r > 1$,
    \[
        \|W_\star\|_\infty \le \begin{cases}
            r(r - 1)^\frac{r + 1}{2} \;\; \text{if} \; G \; \text{is unweighted}, \\
            Z (r - 1)^\frac{r + 1}{2} \;\; \text{if} \; G \; \text{is weighted}.
        \end{cases}
    \]
    Consequently, using Theorems \ref{weighted} and \ref{unweighted}, we get: 
    \[
        \|W_\star\|_\infty \le \begin{cases}
            (3n - 2)^\frac{3n - 1}{2} \;\; \text{if} \; G \; \text{is unweighted}, \\
            Z(3m)^\frac{3m + 1}{2} \;\; \text{if} \; G \; \text{is weighted}.
        \end{cases}
    \]
\end{theorem}
We defer the proof of this theorem to \refapp{app:proof}. These upper bounds help us show provable optimality of the mixed-integer program using an appropriately large value of $M$. In practice though, it may not make sense to use large strength matrices $W$ due to instability of the quantum system and potential introduction of noise at high strengths. In our computations therefore, we set this constant as the sum of edge weights in the graph, i.e., $M = \sum_e |z(e)|$ for 8 node graphs, and obtain solutions with strengths much smaller than $M$.

We further note that if the objective is to minimize the $L_1$ norm, then the program becomes linear:
\begin{align}
    \min &\sum_{i} b_i\\
    & W_{i,i} \leq b_i , ~~ \text{for } i \in [k],\\
    & W_{i,i} \geq - b_i, ~~ \text{for } i \in [k],\\
    & (P^TWP)_{i,j} = A_{i,j} ~~ \text{for } i\neq j, \\
    &b_{i} \in \mathbb{R_+}, \text{for } i\in [k], 
\end{align}
which is typically easier to solve compared to the $L_0$ norm minimization and can be solved to optimality without using the big-$M$ constraints. 

\section{Proof for Theorem \ref{thm:upperbound}.}\label{app:proof}

For a real matrix $B$, we define $\|B\|_0 = \sum_{B_{i, j} \neq 0} 1$, $\|B \|_1 = \sum_{i, j} |B_{i, j}|$ and $\|B\|_\infty = \max_{i, j} |B_{i, j}|$, that is, we look at the matrix $B$ as a vector and use the corresponding vector norms. Further, for a square matrix $B \in \{-1, 1\}^{r \times r}$, we use the bound known as Hadamard's inequality \cite{garling_inequalities}:
\begin{equation*}\label{eq: Hadamard}
    \det(B) \le r^{\frac{r}{2}}.
\end{equation*}

Suppose $B$ is invertible. For $r = 1$, $\|B^{-1}\|_\infty = 1$. For $r > 1$,
\begin{equation}\label{eq: inverse-bound}
    \|B^{-1}\|_\infty = \frac{\|\text{adj}(B)\|_\infty}{|\det(B)|} \le (r - 1)^\frac{r - 1}{2}.
\end{equation}
The inequality holds because $\det(B) \ge 1$ (due to invertibility) and the entries of $\text{adj}(B)$ are minors of $B$, and are therefore bounded by Hadamard's inequality above.

We now proceed to prove the theorem. Recall the graph coupling problem:
\[
    \min \|W\|_0: \{P^TWP \odot J = A\},
\]
where $W$ is the diagonal strength matrix, $P$ is a \tsout{pulse} matrix with entries in $\pm 1$, $A$ is the weighted adjacency matrix of the graph to be constructed and $J$ is the coupling matrix given by $J_{i,j}=1$ for $i\neq j$ and $0$ otherwise.

As in the MIP formulation in \refapp{app:MIP}, we can assume that the rows of \tsout{pulse matrix} $P$ are all possible sequences $\{ \pm 1 \}^n$ up to sign so that $P$ has $k = 2^{n - 1}$ rows, and $W$ is a $k \times k$ diagonal matrix. Then, for $i \neq j$, the constraint $\big(P^TWP \odot J\big)_{i, j} = \sum_{a = 1}^k P_{a, i} P_{a, j} W_{a, a} = A_{i, j}$ is linear in strengths $W_{a, a}$. Let us say $Q_{(i, j), a} = P_{a, i} P_{a, j}$ for $a \in [k]$ and $i, j \in [n], i \neq j$. For ease of notation, we treat $W$ and $A$ as vectors and throughout this proof, we use the following reformulation of the problem:
\[
    \min_{W \in \mathbb{R}^{k}} \|W\|_0: QW = A,
\]
where $Q$ is an appropriately defined matrix in $\{\pm 1\}^{\frac{n(n - 1)}{2} \times k}$ and $A \in \R^{\frac{n(n - 1)}{2}}$. As in the statement of the theorem, let $W_\star$ be an optimal strength matrix for this problem and let $r = \| W_\star\|_0$.

Let $Q^\prime$ be the submatrix of $Q$ obtained by removing all columns with indices $j$ for which $(W_\star)_j = 0$, and let $W_\star^\prime$ be the corresponding restriction of $W_\star$. Then, since $Q W_\star = A$, we have $Q^\prime W_\star^\prime = A$, and $Q^\prime \in \{ \pm 1 \}^{\frac{n(n - 1)}{2} \times r}$. We claim that the columns of $Q^\prime$ are linearly independent, and therefore that $\text{rank}(Q^\prime) = r$. To see this, let $C_1, \ldots, C_r$ be the columns of $Q^\prime$. Suppose $C_r = \sum_{i = 1}^{r - 1} \alpha_i C_i$ for some numbers $\alpha_i \in \R$. Define $v \in \R^{r}$ as

\[
    v_j = \begin{cases}
        (W_\star^\prime)_j + \alpha_j (W_\star^\prime)_r & \text{if} \; 1 \le j \le r - 1, \\
        0 & \text{if} \; j = r.
    \end{cases}
\]
Then, $Q^\prime v = A$ and $\|v\|_0 \le r - 1 < r = \|W_\star^\prime \|_0$, a contradiction. Therefore, columns $C_1, \ldots, C_r$ are independent.

\medskip
Since $\text{rank}(Q^\prime) = r$, $Q^\prime$ has a set of $r$ linearly independent rows $\{R_1, \ldots, R_r \}$ such that the restriction of $Q^\prime$ to these rows is an invertible matrix. Call this matrix $Q^{\prime\prime}$ and call the corresponding restriction of $A$ as $A^\prime$. Then, $Q^{\prime\prime} W_\star^\prime = A^\prime$, where $Q^{\prime\prime}$ is an invertible matrix in $\{ \pm 1 \}^{r \times r}$. \refeq{eq: inverse-bound} gives us that for $r > 1$,
\begin{align*}
    \|W_\star\|_\infty &= \|W_\star^\prime\|_\infty \\
    &= \big\|\big(Q^{\prime\prime}\big)^{-1} A^\prime \big\|_\infty \\
    &\le \big\|\big(Q^{\prime\prime}\big)^{-1} \big\|_\infty \times \|A^\prime \|_1 \\
    &\le \|A^\prime \|_1 (r - 1)^{\frac{r - 1}{2}}.
\end{align*}

For unweighted graphs, $\|A^\prime\|_1 = \|A^\prime \|_0 \le r$ and therefore $\|W_\star\|_\infty \le r(r - 1)^\frac{r - 1}{2} \le r^\frac{r + 1}{2}$. For weighted graphs, $\|A^\prime \|_1 \le \|A \|_1 \le Z$, establishing the bound.

The second set of bounds follow using bounds on $r$ in Theorems \ref{weighted} and \ref{unweighted}: $r = \gc(G) \le 3n - 2$ for unweighted graphs and $r \le 3m + 1$ for weighted graphs.

\bibliography{CostFunction}

\end{document}